\documentclass[pdftex,a4paper,11pt]{article} 

\usepackage{fullpage}

\usepackage[latin1]{inputenc}
\usepackage[english]{babel}
\usepackage[T1]{fontenc}
\usepackage{lmodern} 
\usepackage{graphicx}
\usepackage{wrapfig}
\usepackage{subfigure}
\usepackage{captcont}
\usepackage{enumitem}
\usepackage{datetime}

\newcommand{\titel}{Contractions, Removals and How to Certify 3-Connectivity in Linear Time} 

\usepackage{algorithm} 
\usepackage[noend]{algpseudocode} 

\usepackage[usenames]{color} 
\definecolor{hellblau}{rgb}{0.2,0.4,1} 
\definecolor{dunkelblau}{rgb}{0,0,0.8}
\definecolor{dunkelgruen}{rgb}{0,0.5,0}
\usepackage[
	pdftex,
	colorlinks,
	linkcolor=dunkelblau,
	urlcolor=dunkelblau,
	citecolor=dunkelgruen,
	bookmarks=true,
	linktocpage=true,
	pdftitle={\titel},
	pdfauthor={Jens M. Schmidt},
	pdfsubject={},
	pdfkeywords={}%
]{hyperref} 
\usepackage{xspace} 
\newcommand{\BG}{BG\xspace}
\newcommand{\Children}{\emph{Children}_{12}\xspace}
\newcommand{\Type}{\emph{Type}_3\xspace}
\usepackage{setspace} 

\usepackage{amsmath} 
\usepackage{amsthm} 
\usepackage{amsfonts} 
\theoremstyle{plain} 
	\newtheorem{satz}{Satz}[] 
	\newtheorem{theorem}[satz]{Theorem}
	\newtheorem{lemma}[satz]{Lemma}
	
	\newtheorem{proposition}[satz]{Proposition}
\theoremstyle{remark} 
\theoremstyle{definition} 
	\newtheorem{definition}[satz]{Definition}
	\newtheorem{corollary}[satz]{Corollary}

\begin{document}
	\title{\titel}
	\author{
					Jens M. Schmidt%
		\thanks{This research was supported by the Deutsche Forschungsgemeinschaft within the research training group ``Methods for Discrete Structures'' (GRK 1408). Email: \href{mailto:jens.schmidt@inf.fu-berlin.de}{jens.schmidt@inf.fu-berlin.de}.}\\
	Freie Universit\"at Berlin, Germany}
	\date{}
\maketitle

\begin{abstract}

It is well-known as an existence result that every $3$-connected graph $G=(V,E)$ on more than $4$ vertices admits a sequence of contractions and a sequence of removal operations to $K_4$ such that every intermediate graph is $3$-connected. We show that both sequences can be computed in optimal time, improving the previously best known running times of $O(|V|^2)$ to $O(|V|+|E|)$. This settles also the open question of finding a linear time $3$-connectivity test that is certifying and extends to a certifying $3$-edge-connectivity test in the same time. The certificates used are easy to verify in time $O(|E|)$.
\end{abstract}




\section{Introduction}\label{introduction}
The class of $3$-connected graphs has been studied intensively for many reasons in the past $50$ years. One algorithmic reason is that graph problems can often be reduced to handle $3$-connected graphs; applications include problems in graph drawing (see~\cite{Mutzel2003} for a survey), problems related to planarity~\cite{Bertolazzi1998,Gutwenger2001a} and online problems on planar graphs (see~\cite{Battista1990} for a survey). From a complexity point of view, $3$-connectivity is in particular important for problems dealing with longest paths, because it lies, somewhat surprisingly, on the borderline of NP-hardness: Finding a Hamiltonian cycle is NP-hard for $3$-connected planar graphs~\cite{Garey1976} but becomes solvable in linear running time~\cite{Chiba1989} for higher connectivities, as $4$-connected planar graphs have been proven to be Hamiltonian~\cite{Tutte1956}.


We want to design efficient algorithms from inductively defined constructions of graph classes. In general, such constructions start with a set of base graphs and apply iteratively operations from a fixed set of operations such that precisely the members of the graph class of interest are constructed. This way we obtain a (not necessarily unique) sequence of graphs for each member $G$ of the graph class, which we call a \emph{construction sequence} of $G$. The construction does not only give a computational approach to test membership in these classes, it can also be exploited to prove properties of the graph class using only the fixed set of operations applied in every step. Fortunately, graph theory provides inductively defined constructions for many graph classes, including planar graphs, triangulations, $k$-connected graphs for $k \leq 4$, regular graphs and various intersections of these classes~\cite{Batagelj1986,Batagelj1989,Johnson1963}. However, most of these constructions have not been exploited computationally.

For the class of $3$-connected graphs, one of the most noted constructions is due to Tutte~\cite{Tutte1961}, based on the following fact: Every $3$-connected graph $G$ on more than $4$ vertices contains a \emph{contractible} edge, i.\,e., an edge that preserves the graph to be $3$-connected upon contraction. Contracting iteratively this edge yields a sequence of $3$-connected graphs top-down from $G$ to a $K_4$-multigraph. Unfortunately, also non-$3$-connected graphs can contain contractible edges, but adding a side condition establishes a full characterization: A graph $G$ on more than $4$ vertices is $3$-connected if and only if there is a \emph{sequence of contractions} from $G$ to a $K_4$-multigraph on edges $e$ with both end vertices having at least 3 neighbors~\cite{Elmasry}. Every contracted edge in this sequence is then contractible. It is also possible to describe this sequence bottom-up from $K_4$ to $G$ by using the inverse operations edge addition and vertex splitting; in fact this is the original form as stated in Tutte's famous wheel theorem~\cite{Tutte1961}.


Barnette and Gr\"unbaum~\cite{Barnette1969} and Tutte~\cite{Tutte1966} give a different construction of $3$-connected graphs that is based on the following argument: Every $3$-connected graph $G \neq K_4$ contains a \emph{removable} edge. \emph{Removing} this edge leads, similar as in the sequence of contractions, to a top-down construction sequence from $G$ to $K_4$. Adding a side condition then fully characterizes $3$-connected graphs. We will define \emph{removals} and \emph{removable} edges in Section~\ref{constructionSequences}. Again, the original proposed construction was given bottom-up from $K_4$ to $G$, using three operations.

Although both existence theorems on contractible and removable edges are used frequently in graph theory~\cite{Thomassen1981,Thomassen2006,Tutte1966}, the first non-trivial computational results to create the corresponding construction sequences were published more than $45$ years afterwards: In $2006$, Albroscheit~\cite{Albroscheit2006} gave an algorithm that computes a construction sequence for $3$-connected graphs in $O(|V|^2)$ time in which contractions and removals are allowed to intermix. In $2010$, an algorithm was given~\cite{Schmidt2010} that constructs the (pure) sequences of contractions and removals, respectively, in the same running time. One of the building blocks of this algorithm is a straight-forward transformation from the sequence of removals to the sequence of contractions in time $O(|E|)$. This shows that the sequence of Barnette and Gr\"unbaum is algorithmically at least as powerful as the sequence of contractions. It is important to note that all algorithms do not rely on the $3$-connectivity test of Hopcroft and Tarjan~\cite{Hopcroft1973}, which runs in linear time but is rather involved. It was also shown that all previously mentioned construction sequences can be stored in linear space $O(n+m)$~\cite{Schmidt2010}. Nevertheless, we are not aware of any algorithm that computes any of these sequences in subquadratic time up to now.

The main contribution of this paper is an optimal algorithm that computes the construction sequence of Barnette and Gr\"unbaum bottom-up in time and space $O(|V|+|E|)$. This has a number of consequences.

\paragraph{Top-down and bottom-up variants of both constructions.}
One can immediately obtain the sequence of removals out of Barnette and Gr\"unbaum's construction sequence by replacing every operation with its inverse removal operation. Applying the transformation of~\cite{Schmidt2010} implies optimal time and space algorithms for the sequence of contractions and its bottom-up variant as well.

\paragraph{Certifying 3-connectivity in linear time.}
Blum and Kannan~\cite{Blum1989} initiated the concept of programs that check their work. Mehlhorn and N\"aher~\cite{Kratsch2006,Mehlhorn1998,Mehlhorn1999a} developed this idea further and introduced the concept of \emph{certifying algorithms}, which give a small and easy-to-verify certificate of correctness along with their output.
Achieving such algorithms is a major goal for problems where the fastest solutions known are complicated and difficult to implement. Testing a graph on $3$-connectivity is such a problem, but surprisingly few work has been devoted to certify $3$-connectivity, although sophisticated linear-time recognition algorithms (not giving an easy-to-verify certificate) are known for over $35$ years~\cite{Hopcroft1973,Vo1983,Vo1983a}. The currently fastest algorithms that certify $3$-connectivity need $O(|V|^2)$ time and use construction sequences as certificates~\cite{Albroscheit2006,Schmidt2010}. Recently, a linear time certifying algorithm for $3$-connectivity has been given for the subclass of Hamiltonian graphs, when the Hamiltonian cycle is part of the input~\cite{Elmasry}. In general, finding a certifying algorithm for $3$-connectivity in subquadratic time is an open problem~\cite{Elmasry}.

We give a linear-time certifying algorithm for $3$-connectivity by using Barnette and Gr\"unbaum's construction sequence as certificate. The certificate can be easily verified in time $O(|E|)$, as shown in~\cite{Schmidt2010}. This implies also a new, simple-to-implement and certifying test on $3$-connectivity in linear time and space that is path-based and neither relies on the algorithm of Hopcroft and Tarjan nor uses low-points.

\paragraph{Certifying 3-edge-connectivity in linear time.}
We are not aware of any test for $3$-edge-connectivity that is certifying and runs in linear time. Galil and Italiano~\cite{Galil1991} show that testing $k$-edge-connectivity of a graph $G$ can be reduced to test $k$-vertex-connectivity on a slightly modified graph $G'$, blowing up the number of vertices and edges only by a factor $O(k)$. For $k=3$, the reduction blows up each vertex $v \in V(G)$ to a wheel graph with as many spokes as $v$ has neighbors. Constructing $G'$ and applying the certifying $3$-vertex-connectivity test to $G'$ yields a certifying $3$-edge-connectivity test in linear time and space. However, we have to augment the certificate by the mapping $\phi$ that maps each vertex of $G$ to the vertices and edges being contained in the corresponding wheel graph in $G'$. This ensures that the construction of $G'$ can be verified while preserving linear time and space.
\newline


\section{Construction Sequences}\label{constructionSequences}
Let $G=(V,E)$ be a finite graph with $n$ vertices and $m$ edges.
For $k \geq 1$, a graph $G$ is \emph{$k$-connected} if $n > k$ and deleting every set of $k-1$ vertices leaves a connected graph. A vertex (a pair of vertices) that leaves a disconnected graph upon deletion is called a \emph{cut vertex} (a \emph{separation pair}).
Let $v \rightarrow_G w$ denote a path $P$ from vertex $v$ to vertex $w$ in $G$ and let $s(P):=v$ and $t(P):=w$. For a vertex $v$ in $G$, let $N(v) = \{w \mid vw \in E\}$ denote its set of neighbors and $deg(v)$ its degree. Let $\delta(G)$ be the minimum degree in $G$.

Let $K_n$ be the complete graph on $n$ vertices and let $K_n^m$ be the complete graph on $n$ vertices with $m$ edges between each pair of vertices. For a rooted tree $T$ and $x \in V(T)$, let $T(x)$ be the maximal subtree of $T$ rooted at $x$. We assume for convenience that the input graph $G$ is simple for the rest of the paper, although all results extend to multigraphs.
%
A \emph{subdivision} of a graph $G$ replaces each edge of $G$ by a path of length at least one. Conversely, we want a notation to get back to the graph without subdivided edges. If $deg(v)=2$ and $|N(v) \setminus \{v\}|=2$, let \emph{smoothing} $v$ delete $v$ followed by adding an edge between its neighbors. If one of the conditions is violated, let \emph{smoothing} $v$ not change the graph.

\emph{Removing} an edge $e=xy$ of a graph deletes $e$ followed by smoothing $x$ and $y$. An edge of $G$ is called \emph{removable}, if removing it results in a $3$-connected graph. Iteratively removing removable edges in a $3$-connected graph $G$ leads to a \emph{sequence of removals} from $G$ to $K_4$, the existence of which characterizes $3$-connected graphs when adding a side condition similar as in the sequence of contractions. We describe the equivalent bottom-up construction of $G$ due to Barnette and Gr\"unbaum. The construction starts with $K_4$ and applies iteratively one of the three following operations, which are called \emph{\BG-operations} (see Figure~\ref{fig:BGOperations}):

1. Add an edge $xy$ (possibly a parallel edge).

2. Subdivide an edge $ab$ by a vertex $x$ and add the edge $xy$ for a vertex $y \notin \{a,b\}$.

3. Subdivide two non-parallel edges by vertices $x$ and $y$, respectively, and add the edge $xy$.

\vspace{11pt}

\begin{figure}[htb]
	\centering
	\subfigure[parallel edges allowed]{
		\includegraphics[scale=0.5]{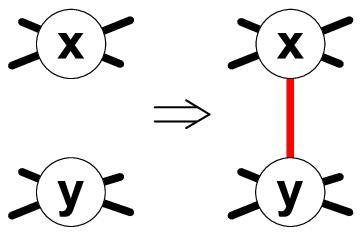}
		\label{fig:BG1}
	}
	\hfill
	\subfigure[$y$, $a$ and $b$ pairwise distinct]{
		\includegraphics[scale=0.5]{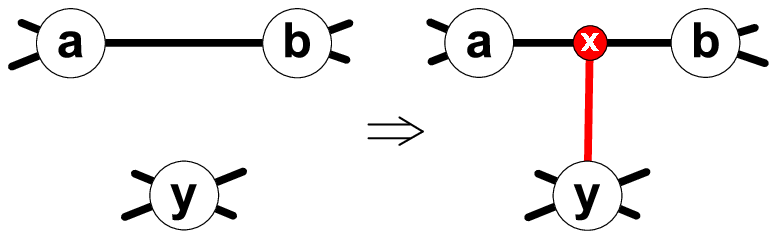}
		\label{fig:BG2}
	}
	\hfill
	\subfigure[$e$ and $f$ neither identical nor parallel]{
		\includegraphics[scale=0.5]{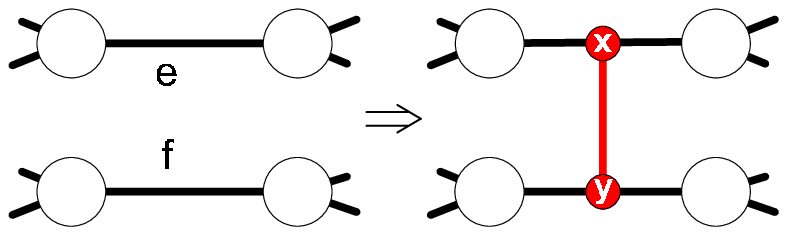}
		\label{fig:BG3}
	}
	\caption{The three \BG-operations.}
	\label{fig:BGOperations}
\end{figure}

Let $G_4,G_5,\ldots,G_z$ with $G_4 = K_4$ and $G_z = G$ be a construction sequence of $G$ using \BG-operations. As $K_4$ is $3$-connected and \BG-operations preserve $3$-connectivity, every $G_l$ with $4 < l \leq z$ is also $3$-connected. We represent the construction sequence in a different, but equivalent way, as shown in~\cite{Barnette1969,Schmidt2010}: Each graph $G_l$ corresponds to a unique $G_l$-subdivision $S_l$ in $G$, which can be readily seen by iteratively deleting the removable edges in the top-down variant.
In this representation, the vertices of $G_l$ correspond to the vertices in $S_l$ of degree at least $3$; we call the latter \emph{real} vertices.
We define the operations on $S_l$ that correspond to \BG-operations. Let $V_{real}(S_l)$ be the set of real vertices in $S_l$. The \emph{links} of $S_l$ are the unique paths in $S_l$ that contain real end vertices but no other real vertex. Let two links be \emph{parallel} if they share the same end vertices.

\begin{definition}\label{bgpathdefinition}
A \emph{\BG-path} for a subgraph $S_l \subset G$ is a path $P = x \rightarrow_G y$ with the properties:
\begin{enumerate}
	\item $S_l \cap P = \{x,y\}$\label{bgpathdefinition1}
	\item Every link of $S_l$ that contains $x$ and $y$, contains them as end vertices.\label{bgpathdefinition2}
	\item If $x$ and $y$ are inner vertices of links $L_x$ and $L_y$ of $S_l$, respectively, and $|V_{real}(S_l)| \geq 4$, then $L_x$ and $L_y$ are not parallel.\label{bgpathdefinition3}
\end{enumerate}
\end{definition}

It is easy to see that every \BG-path for $S_l$ corresponds to a \BG-operation on $G_l$ and vice versa. The choice of the $K_4$-subdivision $S_4$ is not crucial~\cite{Schmidt2010}: At the expense of having additional parallel edges in intermediate graphs $G_l$, there exists a construction sequence to $G$ from \emph{each} prescribed $K_4$-subdivision in $G$. This provides an efficient computational approach to construction sequences, since we can start with an arbitrary $K_4$-subdivision $S_4$ in $G$. The representation with subdivisions allows then to search the next \BG-path in the neighborhood of the current subdivision in $G$. We summarize the results.


\begin{theorem}\label{characterizations}
The following statements are equivalent:
\begin{align}
	&\ \text{A simple graph $G$ is $3$-connected}\label{triconnected}\\
	\Leftrightarrow &\ \exists \text{ sequence of \BG-operations from $K_4$ to $G$ (see~\cite{Barnette1969,Tutte1966})}\label{BGConstruction}\\
	\Leftrightarrow &\ \exists \text{ sequence of \BG-paths from a $K_4$-subdivision in $G$ to $G$ and $\delta(G) \geq 3$ (see~\cite{Barnette1969,Schmidt2010})}\label{BGPathConstructionSingle}\\
	\Leftrightarrow &\ \exists \text{ sequence of \BG-paths from each $K_4$-subdivision in $G$ to $G$ and $\delta(G) \geq 3$ (see~\cite{Schmidt2010})}\label{BGPathConstruction}\\
	\Leftrightarrow &\ \exists \text{ sequence of removals from $G$ to $K_4$ on removable edges $e=xy$}\notag\\
	&\ \text{with $|N(x)| \geq 3$, $|N(y)| \geq 3$ and $|N(x) \cup N(y)| \geq 5$ (see~\cite{Schmidt2010})}\label{removals}\\
	\Leftrightarrow &\ \exists \text{ sequence of contractions from $G$ to $K_4$ on edges $e=xy$ with $|N(x)| \geq 3$}\notag\\
	&\ \text{ and $|N(y)| \geq 3$ (see~\cite{Tutte1961})}\label{contractions}
\end{align}
\end{theorem}

The following Lemma allows to focus only on computing sequence~\ref{characterizations}.\eqref{BGPathConstructionSingle}.

\begin{lemma}[{\cite[Proof of Theorem~2 and Lemma~4.1]{Schmidt2010}}]\label{transformation}
There is an algorithm that transforms a given sequence~\emph{\ref{characterizations}}.\eqref{BGConstruction}, \emph{\ref{characterizations}}.\eqref{BGPathConstructionSingle} or~\emph{\ref{characterizations}}.\eqref{removals} to each of the sequences~\emph{\ref{characterizations}}.\eqref{BGConstruction}-\emph{\ref{characterizations}}.\eqref{contractions} in linear time. If the transformation yields one of the sequences~\emph{\ref{characterizations}}.\eqref{BGConstruction}-\emph{\ref{characterizations}}.\eqref{removals}, the number of operations is preserved.
\end{lemma}

Throughout the rest of the paper, a \emph{construction sequence} will therefore refer to sequence~\ref{characterizations}.\eqref{BGPathConstructionSingle}, unless stated otherwise. To construct such a sequence, we will use the following Lemma.

\begin{lemma}[\cite{Schmidt2010}]\label{multipleconstruction}
Let $G$ be a $3$-connected graph and $H \subset G$ with $H$ being a subdivision of a $3$-connected graph. Then there is a \BG-path for $H$ in $G$. Moreover, every link of $H$ of length at least $2$ contains an inner vertex on which a \BG-path for $H$ starts.
\end{lemma}

Every contraction sequence~\ref{characterizations}.\ref{contractions} contains exactly $n-4$ contractions, implying that the $K_4$-multigraph contains exactly $m-n-2$ parallel edges. The number of removals is also fixed.

\begin{lemma}\label{number}
Every sequence~\ref{characterizations}.\ref{BGConstruction}-\ref{characterizations}.\ref{removals} contains exactly $m-n-2$ operations, i.\,e., $z=m-n+2$.
\end{lemma}
\begin{proof}
It suffices to show the claim for each sequence~\ref{characterizations}.\ref{BGConstruction} with Lemma~\ref{transformation}. Let $a$, $b$ and $c$ denote the number of \BG-operations in the sequence that create zero, one and two new vertices, respectively. Then $b+2c=n-4$ and $a+2b+3c=m-6$ hold, since $K_4$ consists of four vertices and six edges. Subtracting the equations gives that $a+b+c=m-n-2$.
\end{proof}

\section{Chain Decompositions}\label{Chains}
Let $G$ be a $3$-connected graph. According to Lemma~\ref{multipleconstruction}, it suffices to add iteratively \BG-paths to an arbitrary $K_4$-subdivision in $G$ to get a construction sequence. Note that we cannot make wrong decisions when choosing a \BG-path, since Lemma~\ref{multipleconstruction} can always be applied on the new subdivision and therefore ensures a completion of the sequence. Instead of starting with a $K_4$-subdivision, we will w.\,l.\,o.\,g.\ start with a $K_2^3$-subdivision $S_3$ and find a \BG-path for $S_3$ that results in a $K_4$-subdivision. We first show how $S_3$ is computed and then describe a decomposition of $G$ into special paths that allows us to find the \BG-paths efficiently.

A Depth First Search (DFS) is performed on $G$, assigning a Depth First Index (DFI) to every vertex. Let $T$ be the corresponding DFS-tree, $r$ be the root of $T$ and $u$ be the vertex that is visited second. Both, $r$ and $u$, have exactly one child in $T$, as otherwise they would form a separation pair in $G$. For two vertices $v$ and $w$ in $T$, let $v$ be a (\emph{proper}) \emph{ancestor} of $w$ and $w$ be a (\emph{proper}) \emph{descendant} of $v$ if $v \in V(r \rightarrow_T w)$ (and $v \neq w$). A \emph{backedge} is an edge $vw \in E(G) \setminus E(T)$ oriented from $v$ to $w$ with $v$ being an ancestor of $w$ (note that this differs from standard notation). A backedge $vw$ \emph{enters} a subtree $T'$ of a tree if $v \notin V(T')$ but $w \in V(T')$.

To compute $S_3$, we choose two backedges $ra$ and $rb$ and denote the least common ancestor of $a$ and $b$ in $T$ with $x$. The paths $x \rightarrow_T r$, $ra \cup a \rightarrow_T x$ and $rb \cup b \rightarrow_T x$ are the three subdivided edges of $S_3$ in $G$ with real vertices $r$ and $x$. Now, $G$ is decomposed into special paths $\{C_0,C_1,\ldots,C_{m-n+1}\} =: C$, called \emph{chains}, whose edge sets partition $E(G)$. The decomposition imposes a total order $<$ on $C$ with $C_0 < C_1 < \ldots < C_{m-n+1}$ that is identical to the order in which the chains were computed.
We set $C_0 := x \rightarrow_T r$, $C_1 := ra \cup a \rightarrow_T x$ and $C_2 := rb \cup b \rightarrow_T x$. The remaining chains are then computed by applying the following procedure subsequently for each vertex $v$ in increasing DFI-order: For every backedge $vw$ not in a chain, we traverse the path $w \rightarrow_T r$ until a vertex $x$ is found that is already contained in a chain. The traversed path $v \rightarrow_G x$ including $vw$ forms the new chain.

Note that every chain $C_i \neq C_0$ contains exactly one backedge and thus $|C|=m-n+2$. Also, $s(C_i)$ is always a proper ancestor of $t(C_i)$. Chains admit the following tree structure.

\begin{definition}
Let the \emph{parent of a chain} $C_i \neq C_0$ be the chain that contains the edge from $t(C_i)$ to the parent of $t(C_i)$ in $T$.
\end{definition}

\begin{lemma}
The parent relation defines a tree $U$ with $V(U)=C$ and root $C_0$.
\end{lemma}
\begin{proof}
Let $D_0 \neq C_0$ be a chain in $C$ and let $D_1,\ldots,D_k$ be the sequence of chains containing the edges of $t(D_0) \rightarrow_T r$ in that order, omitting double occurrences. By definition of the parent relation, each $D_i$, $0 \leq i < k$, has parent $D_{i+1}$. It follows with $D_k = C_0$ that $U$ is connected. Moreover, $U$ is acyclic, as parent chains are always smaller in $<$ than their children by definition of the decomposition.
\end{proof}

\subsection{Classifying Chains and Restrictions}\label{ClassifyingChains}
We extend the chain decomposition to assign one of the types $1$, $2a$, $2b$, $3a$ and $3b$ to each chain in $C \setminus \{C_0\}$. The motivation for this classification is that chains of certain types are, under some conditions, \BG-paths and therefore allow to compute the next step of the construction sequence. The types are defined by Algorithm~\ref{alg:classify}: E.\,g., a chain $C_i$ with parent $C_k$ is of type~$1$ if $t(C_i) \rightarrow_T s(C_i) \subseteq C_k$ and of type~$2$ if it is not of type~$1$ and $s(C_i)=t(C_k)$ holds (see Figures~\ref{fig:Classification} and~\ref{fig:Example1}). All chains are unmarked at the beginning of Algorithm~\ref{alg:classify}. It is not difficult to show that the decomposition and classification can be carried out in linear time. We omit a proof.

\begin{wrapfigure}[12]{R}{3cm}
	\vspace{-0.99cm}
	\centering
	\includegraphics[scale=0.7]{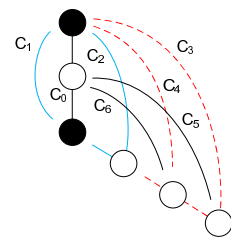}
	\caption{$C_1$ and $C_2$ are of type~$1$, $C_3$ is of type~$2b$, $C_4$ of type~$2a$, $C_5$ of type~$3b$ and $C_6$ of type~$3a$.}
	\label{fig:Classification}
\end{wrapfigure}

\begin{algorithm}
\caption{classify$(C_i \in C \setminus \{C_0\}, \textrm{DFS-tree } T)$}\label{alg:classify}
\begin{algorithmic}[1]
	\State $C_k := \emph{parent}(C_i)$\Comment{the parent of $C_i$ in $U$: $C_k < C_i$}
	\If{$t(C_i) \rightarrow_T s(C_i)$ is contained in $C_k$}\Comment{type~$1$~}\label{typeone}
		\State assign type~$1$ to $C_i$
	\ElsIf{$s(C_i)=s(C_k)$}\Comment{type~$2$: $C_k \neq C_0$, $t(C_i)$ is inner vertex of $C_k$}
		\If{$C_i$ is a backedge}
			\State assign type~$2a$ to $C_i$\Comment{type~$2a$}
		\Else
			\State assign type~$2b$ to $C_i$; mark $C_i$\Comment{type~$2b$}
		\EndIf
	\Else \Comment{type~$3$: $s(C_i) \neq s(C_k)$, $C_k \neq C_0$, $t(C_i)$ is inner vertex of $C_k$}
		\If{$C_k$ is not marked}
			\State assign type~$3a$ to $C_i$\Comment{type~$3a$}
		\Else \Comment{$C_k$ is marked}
			\State assign type~$3b$ to $C_i$; create a list $L_i = \{C_i\}$; $C_j := C_k$\Comment{type~$3b$}\label{assign}
			\While{$C_j$ is marked}\Comment{$L_i$ is called a \emph{caterpillar}}
				\State unmark $C_j$; append $C_j$ to $L_i$; $C_j := \emph{parent}(C_j)$\label{unmarking}
			\EndWhile
		\EndIf
	\EndIf
\end{algorithmic}
\end{algorithm}

\begin{lemma}\label{ChainDecomposition}
Computing a chain decomposition of a $3$-connected graph and classifying each chain with Algorithm~\ref{alg:classify} takes running time $O(n+m)$.
\end{lemma}
%

\begin{definition}\label{upwardsclosed}
Let a subdivision $S_l \subseteq G$ be \emph{upwards-closed} if for each vertex in $S_l$ the edge to its parent is in $E(S_l)$. Let $S_l$ be \emph{modular} if $S_l$ is the union of chains.
\end{definition}

In order to find \BG-paths efficiently, we want to restrict every subdivision $S_l$ to be upwards-closed and modular. However, configurations exist where no \BG-path for a subdivision $S_l$ is a chain, e.\,g., the subdivision $S_3 = \{C_0,C_1,C_2\}$ in Figure~\ref{fig:Classification}. This violates the modularity of $S_{l+1}$ and we have to weaken the restriction: We will allow intermediate subdivisions that are neither upwards-closed nor modular but demand in these cases that we can find a set of $t$ \BG-paths in advance that restores these properties after $t$ steps.

We impose the additional restriction that each link of $S_l$ that consists only of tree edges has no parallel link, except $C_0$ in $S_3$. This prevents \BG-path candidates from violating property~\ref{bgpathdefinition}.\ref{bgpathdefinition3} due to the DFS-structure. We summarize the restrictions.

\begin{enumerate}[label=($R_\arabic*$), leftmargin=*]
	\item For each upwards-closed and modular subdivision $S_l$, \BG-paths are only added as\label{R1} single chains of type~$1$, $2a$ or $3a$, with $S_{l+1}$ being upwards-closed and modular or as sets of $t > 1$ subsequent \BG-paths constructing an upwards-closed modular subdivision $S_{l+t}$ that differs from $S_l$ in exactly $t$ chains of types $2b$ and $3b$.
	\item For each upwards-closed and modular subdivision $S_l$, the links of $S_l$ that consist only of tree edges of $T$ have no parallel links, except $C_0$ in $S_3$.\label{R2}
\end{enumerate}


We refer to the current upwards-closed and modular subdivision in a construction sequence that is restricted by~\ref{R1} and~\ref{R2} as $S^R_l$. The existence of a restricted sequence is shown in Section~\ref{caterpillars}.
We show that chains of type~$3a$ help to find \BG-paths efficiently (proof omitted).

\begin{lemma}\label{structural3a}
Let $C_i$ be a chain of type~$3a$ and $C_k$ the parent of $C_i$ such that $C_k$ but not $C_i$ is contained in $S^R_l$. Then $C_i$ is a \BG-path for $S^R_l$ preserving~\ref{R1} and~\ref{R2}.
\end{lemma}

\begin{wrapfigure}[14]{R}{5.7cm}
	\vspace{-0.5cm}
	\centering
	\subfigure[A \emph{bad} caterpillar $L_j$ with parent $C_k$.]{
		\includegraphics[scale=0.5]{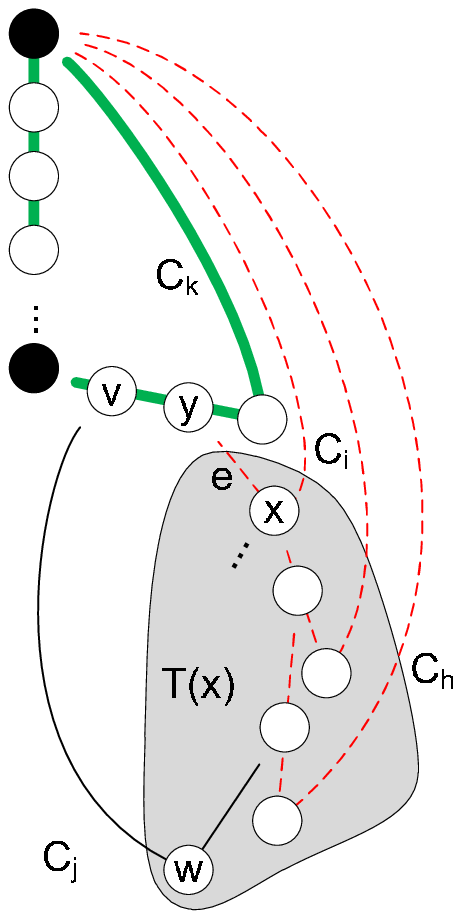}
		\label{fig:CaterpillarBad}
	}
	\hspace{0.01cm}
	\subfigure[A \emph{good} caterpillar $L_j$ with parent $C_k$.]{
		\includegraphics[scale=0.5]{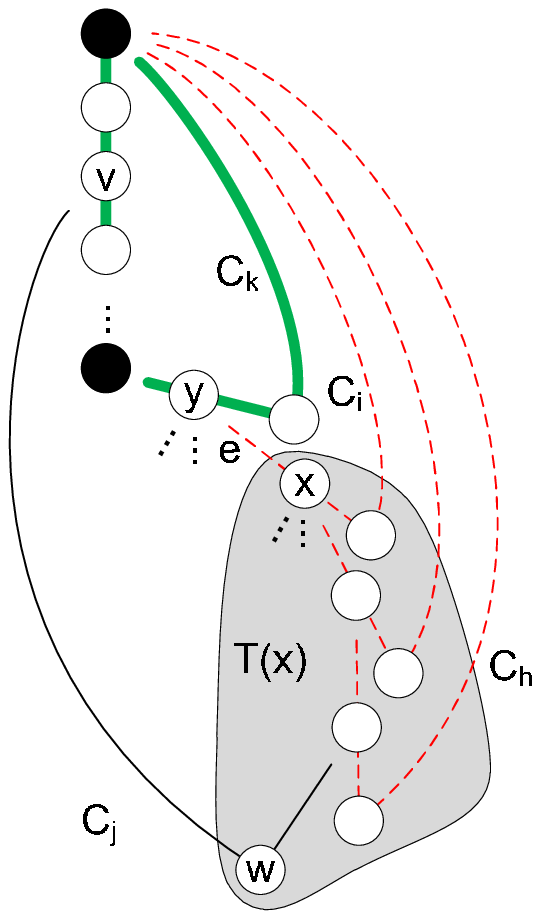}
		\label{fig:CaterpillarGood}
	}
	\caption{Kinds of caterpillars.}
	\label{fig:Caterpillar}
\end{wrapfigure}

\subsection{Caterpillars and Existence of the Restricted Sequence}\label{caterpillars}
While chains of type~$3a$ form \BG-operations under the conditions of Lemma~\ref{structural3a}, chains of types $1$ and $2$ in general do not. For every chain $C_i$ of type~$3b$, Algorithm~\ref{alg:classify} collects a list $L_i$ of chains that contains only $C_i$ and chains of type~$2b$ (see line~\ref{unmarking}). We call each list $L_i$ a \emph{caterpillar}.

\begin{definition}
Let the \emph{parent of a caterpillar} $L_j$ be the parent of the chain in $L_j$ that is minimal with respect to $<$.
Let a caterpillar $L_j$ with parent $C_k$ be \emph{bad} for a subdivision $S_l$ if $s(C_j)$ is a descendant of $t(C_k)$ and $s(C_k) \rightarrow_{C_k} s(C_j)$ contains no inner real vertex (see Figure~\ref{fig:CaterpillarBad}). Otherwise, $L_j$ is called a \emph{good} caterpillar (see Figure~\ref{fig:CaterpillarGood}).
\end{definition}

Caterpillars bundle the single chains of type~$2b$, which cannot immediately be added as \BG-paths. They also offer a simple decomposition into successive \BG-paths.

\begin{lemma}\label{AddCaterpillar}
Let $L_j$ be a caterpillar that consists of $t$ chains and has parent $C_k$. Let $C_k$ but no chain in $L_j$ be contained in $S^R_l$. If $L_j$ is good, $L_j$ can be efficiently decomposed into $t$ successive \BG-paths satisfying~\ref{R1} and~\ref{R2}.
\end{lemma}
\begin{proof} (sketch) Let $y$ be the last vertex of the minimal chain in $L_j$ and let $C_h$ be the parent of $C_j$. We add either the path $P := C_j \cup (t(C_j) \rightarrow_T y)$, followed by $C_i \setminus P$ for all chains $C_i \in L_j \setminus \{C_j\}$ (see Figure~\ref{fig:CaterpillarGood}) or the path $(C_j \cup C_h) \setminus ((t(C_j) \rightarrow_T y) \setminus \{t(C_j)\})$, followed by $t(C_j) \rightarrow_T y$ and $C_i \setminus (t(C_j) \rightarrow_T y)$ for every remaining chain $C_i \in L_j \setminus \{C_j,C_h\}$.
\end{proof}

\begin{definition}\label{segment}
We define the equivalence relation $\sim$ on $E(G) \setminus E(S_l)$ with $e \sim e$ for all $e \in E(G) \setminus E(S_l)$ and with $e \sim f$ for all $e, f \in E(G) \setminus E(S_l)$ if there is a path $e \rightarrow_G f$ without an inner vertex in $S_l$.
Let the \emph{segments} of $S_l$ be the subgraphs of $G$ that are induced by the equivalence classes of $\sim$. Let $H \cap S_l$ be the \emph{attachment points} of $H$.
\end{definition}

\begin{definition}
For a chain $C_i$ and a subdivision $S^R_l$, let $\Children(C_i)$ be the set of children of $C_i$ of types~$1$ and~$2$ that are not contained in $S^R_l$ and let $\Type(C_i)$ be the set of chains of type~$3$ that start at a vertex in $C_i$ and are not contained in $S^R_l$.
\end{definition}

The following theorem is a key result of this paper and leads not only to an existence proof of the restricted construction sequence but also to an efficient algorithm for computing it.

\begin{theorem}\label{correctness}
For a subdivision $S^R_l$, let $C_i$ be a chain such that $\Children(C_j) = \Type(C_j) = \emptyset$ holds for every proper ancestor $C_j$ of $C_i$. Then all chains in $\Children(C_i) \cup \Type(C_i)$ and their proper ancestors that are not already contained in $S^R_l$ can be successively added as \BG-paths (possibly being part of caterpillars) such that~\ref{R1} and~\ref{R2} is preserved. Moreover, the chains in $\Type(C_i)$ that are contained in segments in which the minimal chain is not contained in $\Children(C_i)$ can be added at any point in time in arbitrary order (together with their proper ancestors that are not contained in $S^R_l$).
\end{theorem}

The precondition of Theorem~\ref{correctness} is met in every subdivision: For $S^R_3$, $C_0$ is the desired chain and applying the Theorem on $C_0$ allows to take the descendants of $C_0$ in $U$ in subsequent subdivisions. This ensures the existence of the restricted construction sequence.

\begin{corollary}\label{correctnessCorollary}
Let $G$ be a $3$-connected graph with a chain decomposition $C = \{C_0,\ldots,C_{m-n+1}\}$. Then there is a construction sequence of $G$ restricted by~\ref{R1} and~\ref{R2} that starts with $S^R_3 = \{C_0 \cup C_1 \cup C_2\}$.
\end{corollary}


\section{A Linear-Time Algorithm}\label{algorithm}
With Lemma~\ref{ChainDecomposition}, a chain decomposition, a subdivision $S^R_3$ and the classification of chains can be computed in time $O(n+m)$. Theorem~\ref{correctness} provides an algorithmic method to find iteratively \BG-paths building the restricted construction sequence: Iteratively for each chain $C_i$, $0 \leq i \leq m-n$, we add all chains in $\Children(C_i) \cup \Type(C_i)$ (we say that $C_i$ is \emph{processed}). Note that $C_i$ meets the precondition of Theorem~\ref{correctness} and that $\Children(C_i)$ and $\Type(C_i)$ can be build in time $O(|C_i|+|\Children(C_i)|+|\Type(C_i)|)$ by storing lists of type~$3$ chains at every vertex. We partition the chains in $\Type(C_i)$ into the different segments of $S^R_l$ containing them by storing a pointer on each $C_j \in \Type(C_i)$ to the minimal chain $D$ of the segment containing $C_j$. The chain $D$ is computed by traversing $T$ from $t(C_j)$ to the root until the next vertex is contained in $S^R_l$. The current vertex is then an inner vertex of $D$ (each inner vertex has a pointer to its chain) and we mark each vertex of the traversed path with $D$. Further traversals get $D$ by stopping at the first marker that points to a chain not in $S^R_l$. Since all traversed chains will be added, the running time amortizes to a total of $O(n+m)$.

First, we add all chains in $\Type(C_i)$ that are contained in segments in which the minimal chain is not contained in $\Children(C_i)$ (this can be checked in constant time per chain). According to Theorem~\ref{correctness}, the chains can be added in arbitrary order, as long as their proper ancestors that are not in $S^R_l$ are added before. We want to add the remaining chains in $\Children(C_i) \cup \Type(C_i)$. However, Theorem~\ref{correctness} does not specify in which order the chains have to be added, so we need to compute a valid order on them.

Let $C_j$ be a remaining chain in $\Type(C_i)$ and let $H$ be the segment containing it. Then the minimal chain $D$ in $H$ is of type~$1$ or~$2$, as it is contained in $\Children(C_i)$. If $D$ is of type~$1$ or~$2a$, $s(D)$, $t(D)$, $s(C_j)$ and all other attachment points of $H$ must be contained in $C_i$. The same holds for the remaining case of $D$ being of type~$2b$, as the type~$3b$-chain in the caterpillar containing $D$ cannot start in a proper ancestor of $C_i$ by assumption. Let the \emph{dependent path} of $H$ be the maximal path in $C_i$ connecting two attachment points of $H$, e.\,g., for $D$ being of type~$1$ or~$2a$, the dependent path is just $s(D) \rightarrow_{C_i} t(D)$. We can compute all attachment points of $H$ and therefore the dependent path of $H$ efficiently, as the previous computation provides $D$ and the set of chains $H \cap \Type(C_i)$; we just have to add $s(D)$ and $t(D)$ to the start vertices of the latter chains.

If $D$ is a chain of type~$2a$ (thus, $H=D$) and $t(D)$ is real, we can add $D$. Otherwise, restriction~\ref{R2} implies that every segment $H$ that has a dependent path $P$ without inner real vertices does neither contain chains nor caterpillars forming \BG-paths while preserving~\ref{R1} and~\ref{R2}. Conversely, if $P$ contains an inner real vertex, all chains in $H \cap (\Children(C_i) \cup \Type(C_i))$ can be added: If $D$ is of type~$1$, $D$ does not violate~\ref{R2} and can be added and if $D$ is of type~$2b$, the caterpillar containing $D$ is good and can be added with Lemma~\ref{AddCaterpillar}. As the minimal chain in $H \setminus \{D\}$ is not contained in $\Children(C_i)$, the remaining chains in $H \cap \Type(C_i)$, if exist, can be added as well using Theorem~\ref{correctness}.

\begin{wrapfigure}[16]{R}{7cm}
	\vspace{-0.5cm}
	\centering
	\subfigure[The chain $C_i$ and the chains in $\Type(C_i) \cup \Children(C_i)$.]{
		\includegraphics[scale=0.4]{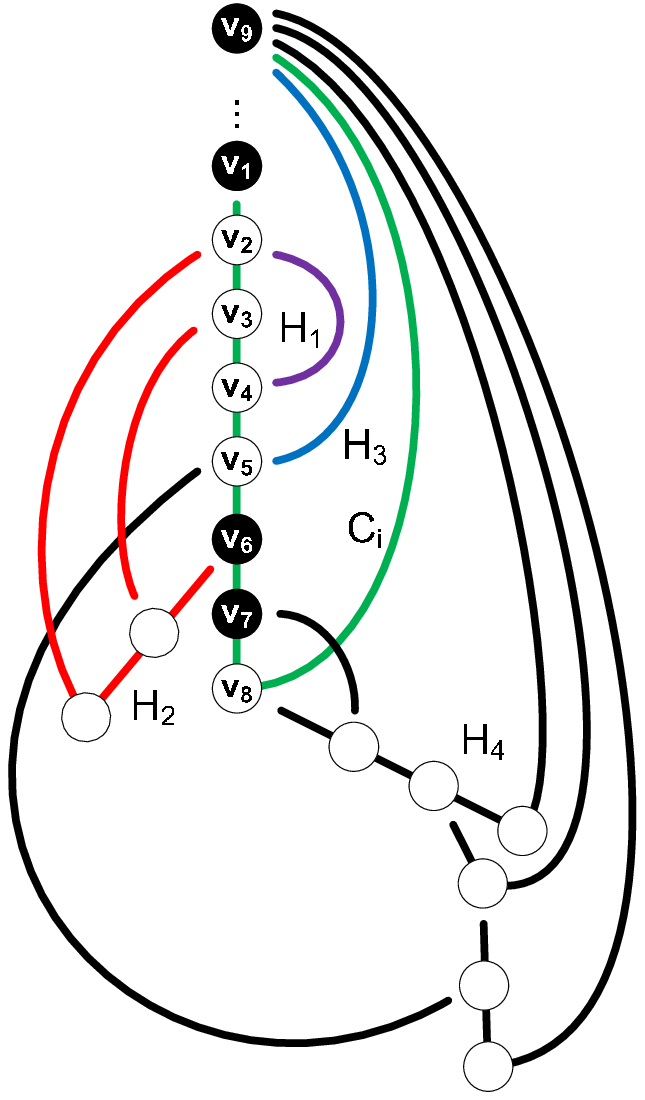}
		\label{fig:Mapping1}
	}
	\hspace{0.1cm}
	\subfigure[The intervals in $I_0$ are constructed from the real vertices in $C_i$.]{
		\includegraphics[scale=0.42]{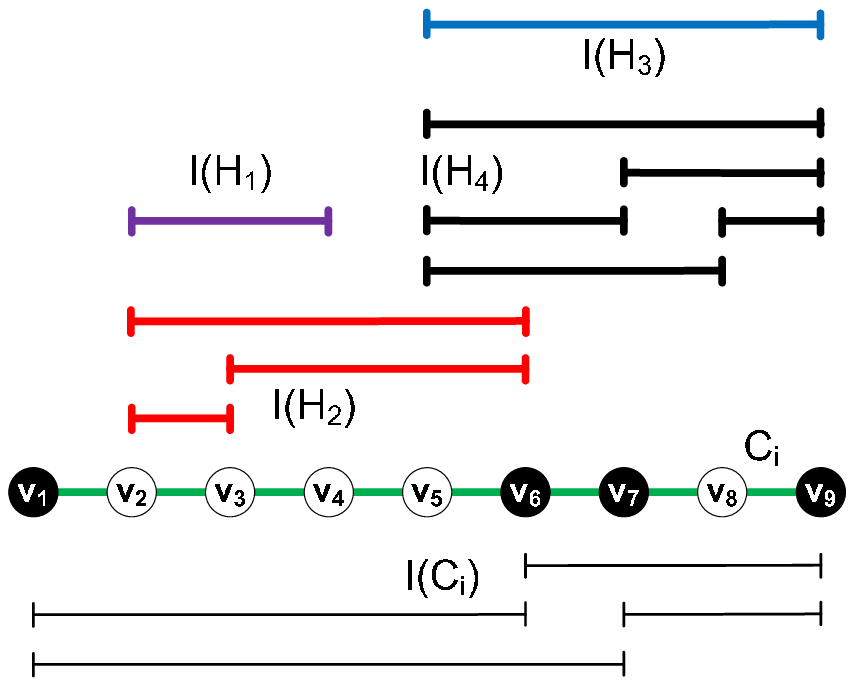}
		\label{fig:Mapping2}
	}
	\caption{Mapping segments in $C_i$. Different shades depict different segments.}
	\label{fig:Mapping}
\end{wrapfigure}

Finding a valid order on the remaining chains in $\Children(C_i) \cup \Type(C_i)$ thus reduces to finding an order on their segments such that the dependent paths of the segments contain inner real vertices. Having this sequence would allow to add subsequently $H \cap (\Children(C_i) \cup \Type(C_i)$) for every segment $H$ in this order. We map each $H$ to a set $I(H)$ of intervals in the range of the dependent path $P$ of $H$: Let $a_1,\ldots,a_k$ be the attachment points of $H$ and let $I(H) := \bigcup_{1 < j \leq k} \{[a_1,a_j]\} \cup \bigcup_{1 < j < k} \{[a_j,a_k]\}$ (see Figure~\ref{fig:Mapping}). Additionally, we map the real vertices $b_1,\ldots,b_k$ of $C_i$ to the set of intervals $I(C_i) := \bigcup_{1 < j < k} \{[b_1,b_j]\} \cup \bigcup_{1 < j < k} \{[b_j,b_k]\}$. This construction is efficient and creates at most $2*(|\Children(C_i)|+|\Type(C_i)|+|V_{real}(C_i)|)$ intervals for $C_i$, which amortizes to a total of $O(n+m)$ for all chains.

Let two intervals $[a,b]$ and $[c,d]$ \emph{overlap} if $a < c < b < d$ or $c < a < d < b$. Starting with $I(C_i)$, we find the next segment with an inner real vertex on its dependent path by finding a next overlapping interval $C_j$ and adding the whole segment that contains $C_j$. This reduction finds the desired order: Clearly, an overlap induces an inner real vertex in the next interval and therefore in the dependent path of the next segment. Conversely, for every segment $H$ with an inner real vertex on its dependent path $P$, an interval can be found that overlaps with $P$, either in $I(C_i)$ if $v$ was real at the beginning or in $I(H')$ for a previously added segment $H'$ (note that segments having only the attachment points $s(C_i)$ and $t(C_i)$ cannot occur, as they contain no chain in $\Children(C_i)$).

A sequence of overlaps from $I(C_i)$ to every other created interval exists if and only if the \emph{overlap graph} (i.\,e., the graph with intervals as vertices and edges between overlapping intervals) is connected. Simple sweep-line algorithms for constructing the connected components of the overlap graph are known~\cite{Olariu1996} (Lemmas~4.1 and 4.2 suffice), run in time $O(t)$ for $t$ intervals and, thus, ensures the efficient computation of the reduction.

\begin{theorem}
The construction sequences \emph{\ref{characterizations}}.\eqref{BGConstruction}, \emph{\ref{characterizations}}.\eqref{BGPathConstruction}, \emph{\ref{characterizations}}.\eqref{removals} and \emph{\ref{characterizations}}.\eqref{contractions} of a $3$-connected graph can be computed in time $O(n+m)$.
\end{theorem}


\begin{figure}
	\centering
	\subfigure[A $3$-connected input graph $G$ with $n=18$ and $m=34$. Straight lines depict the edges of the DFS-tree $T$.]{
		\includegraphics[scale=0.62]{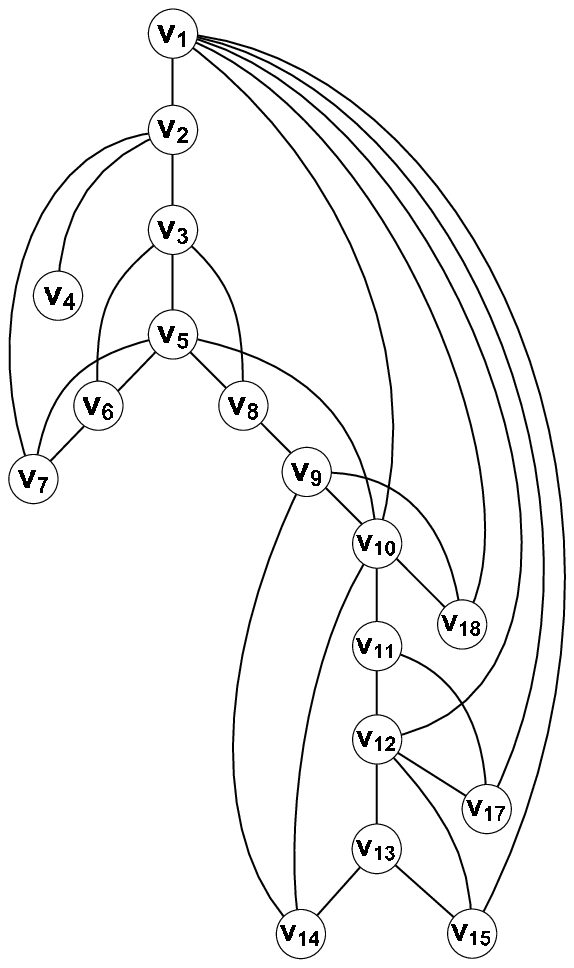}
		\label{fig:Example0}
	}
	\hfill
	\subfigure[The decomposition of $G$ into $m-n+2=18$ chains. Light solid chains are of type~$1$, dashed ones of type~$2$ and black solid ones of type~$3$. The only chain of type~$2a$ is $C_3$. The only chains of type~$3b$ are $C_{14}$ and $C_{16}$, which create the caterpillars $L_{14}=\{C_{14},C_6,C_5\}$ and $L_{16}=\{C_{16},C_4\}$, respectively.]{
		\includegraphics[scale=0.62]{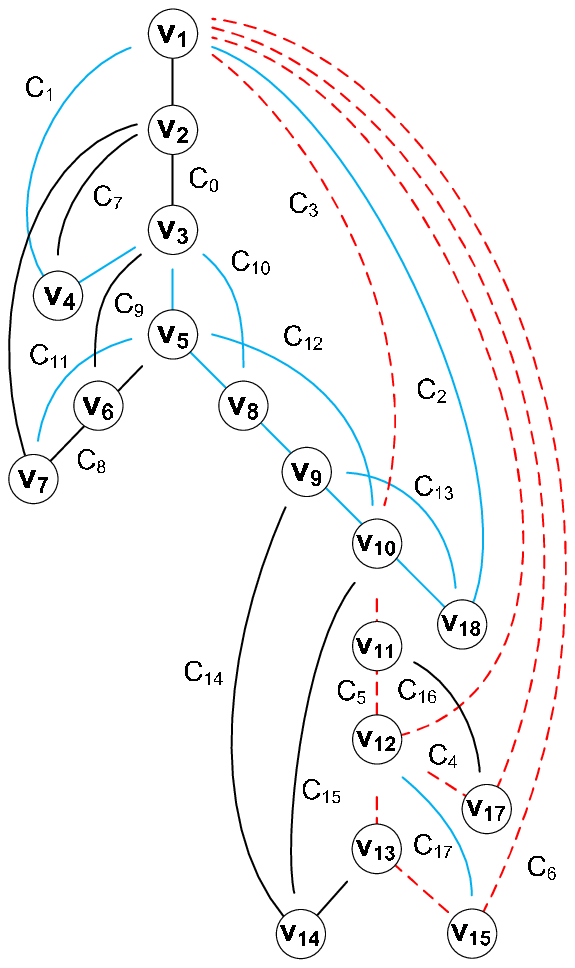}
		\label{fig:Example1}
	}
	\hfill
	\subfigure[The subdivision $S^R_3 = \{C_0,C_1,C_2\}$ (thick edges). We start with processing $C_0$. Since $\Children(C_0) = \emptyset$, we can add all chains in $\Type(C_0)=\{C_7,C_8,C_9\}$. The first two have parents that are already contained in $S^R_3$. We thus add one of them as \BG-path, say $C_7$.]{
		\includegraphics[scale=0.62]{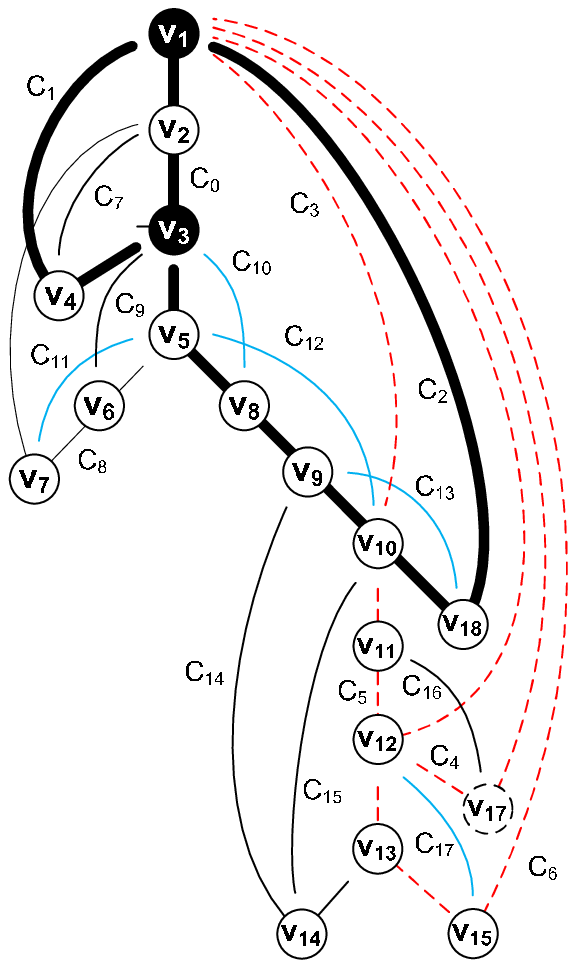}
		\label{fig:Example2}
	}
	\hfill
	\subfigure[The $K_4$-subdivision $S^R_4$. Its real vertices are depicted in black. Note that choosing $C_8$ instead of $C_7$ would have led also to a $K_4$-subdivision. After adding the remaining chain $C_8$ as \BG-path, the parent of $C_9$ is contained in $S^R_5$ and can therefore be added as well.]{
		\includegraphics[scale=0.62]{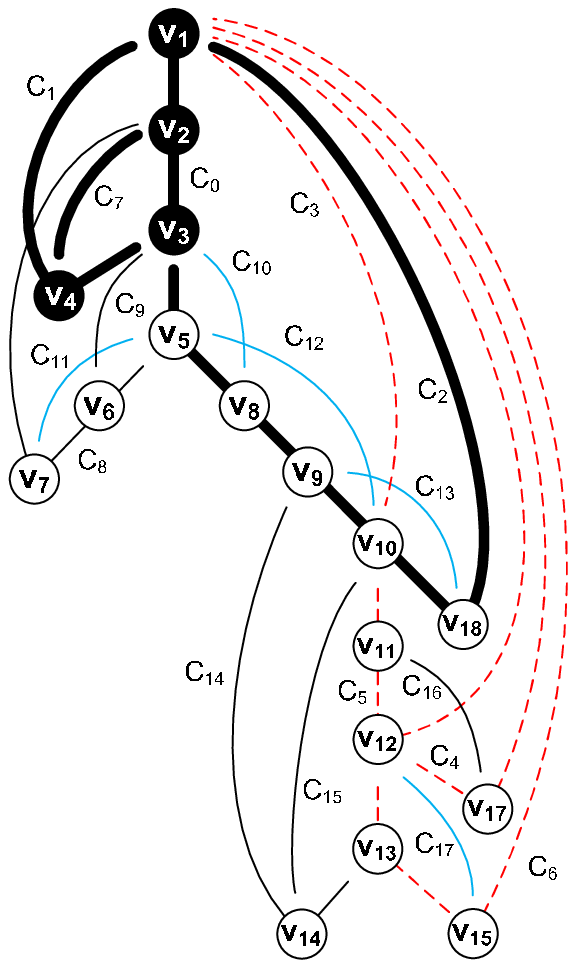}
		\label{fig:Example2_5}
	}
	\subfigure[The subdivision $S^R_6$. We process $C_1$ next, but have to continue to process $C_2$, as there is nothing to add. According to Theorem~\ref{correctness}, $C_3$, $C_{10}$, $C_{11}$, $C_{12}$, $C_{13}$, $C_{15}$ and the caterpillar $L_{14}$ can be added. We first add $C_{11}$, as its segment does not contain a chain in $\Children(C_2)$. To obtain the right order of the remaining chains, we group them by segments.]{
		\includegraphics[scale=0.62]{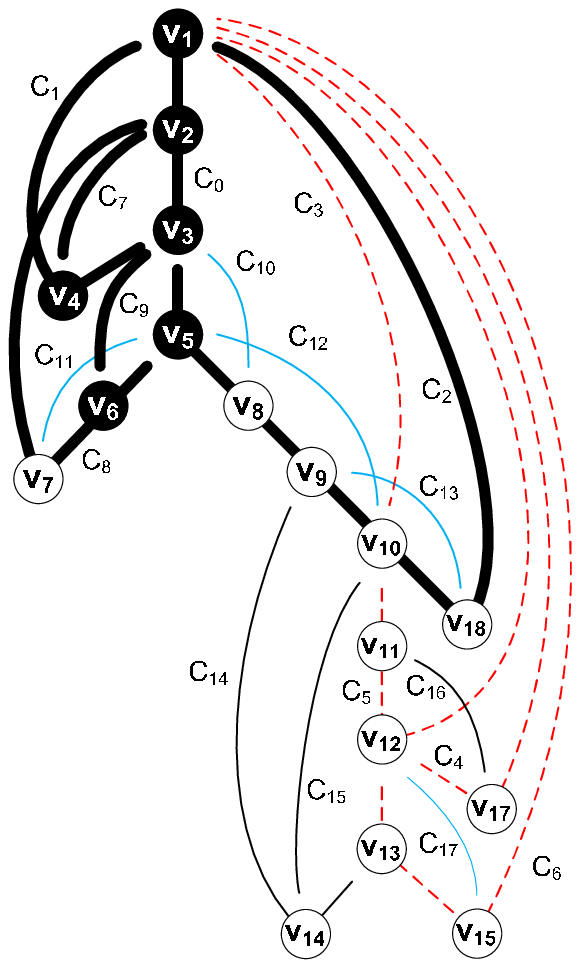}
		\label{fig:Example3}
	}
	\hfill
	\subfigure[We map the segments to intervals. $I(C_2)$ is induced by the real vertices $v_3$, $v_5$ and $v_1$. As $L_{14}$ and $C_{15}$ are in the same segment, they are mapped to the same group of intervals. By overlapping intervals, we get the sequence of segments $I(C_2)$, $I(C_{10})$, $I(C_{12})$, $I(C_{13})$, $I(C_3)$ and $I(L_{14} \cup C_{15})$. Note that overlapped intervals imply adding the whole segment.]{
		\includegraphics[scale=0.55]{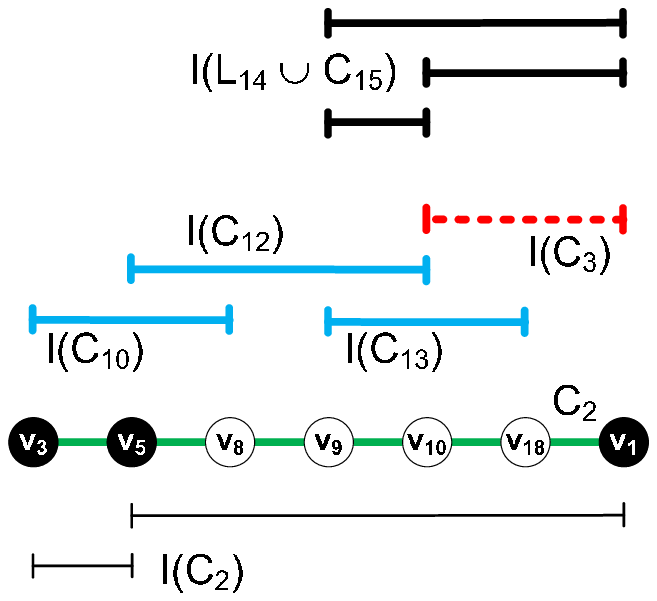}
		\label{fig:Example3_5}
	}
%
	\hfill
	\subfigure[The subdivision $S^R_{14}$. The next non-trivial chain to process is $C_5$. The interval {$[v_{11},v_1] \subset I(L_{16})$} contains the inner real vertex $v_{12}$ and overlaps with {$[v_{10},v_{12}] \subset I(C_5)$}. This implies that $L_{16}$ can be added, forming the two \BG-paths $v_{11} \rightarrow_{G \setminus E(S_{15})} v_1$ and $v_{12} \rightarrow_T v_{17}$.]{
		\includegraphics[scale=0.62]{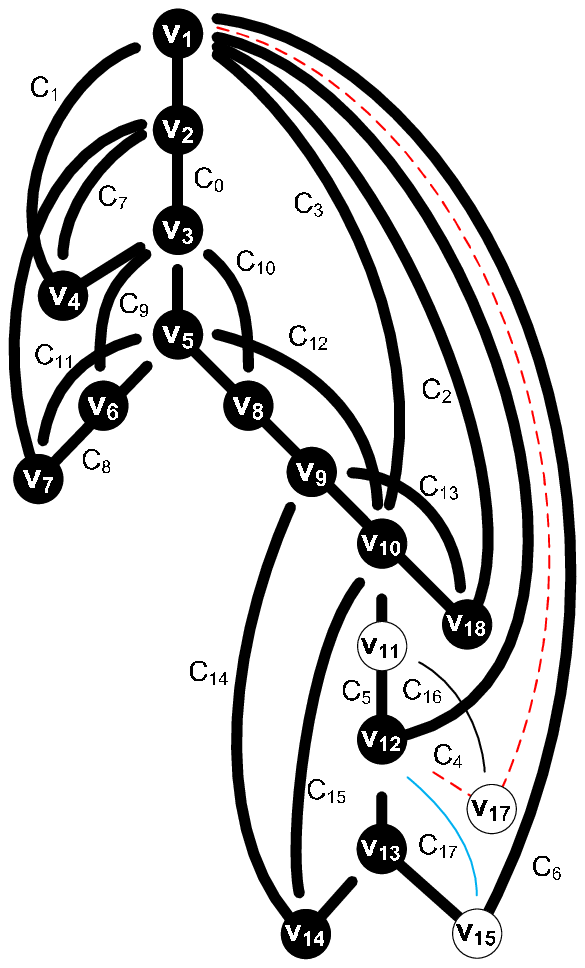}
		\label{fig:Example4}
	}
	\hfill
	\subfigure[The subdivision $S^R_{17}$. It remains to process the chain $C_6$, where $C_{17}$ is added as the last \BG-path of the construction sequence. This results in the subdivision $S^R_{18}$, which is identical to $G$.]{
		\includegraphics[scale=0.62]{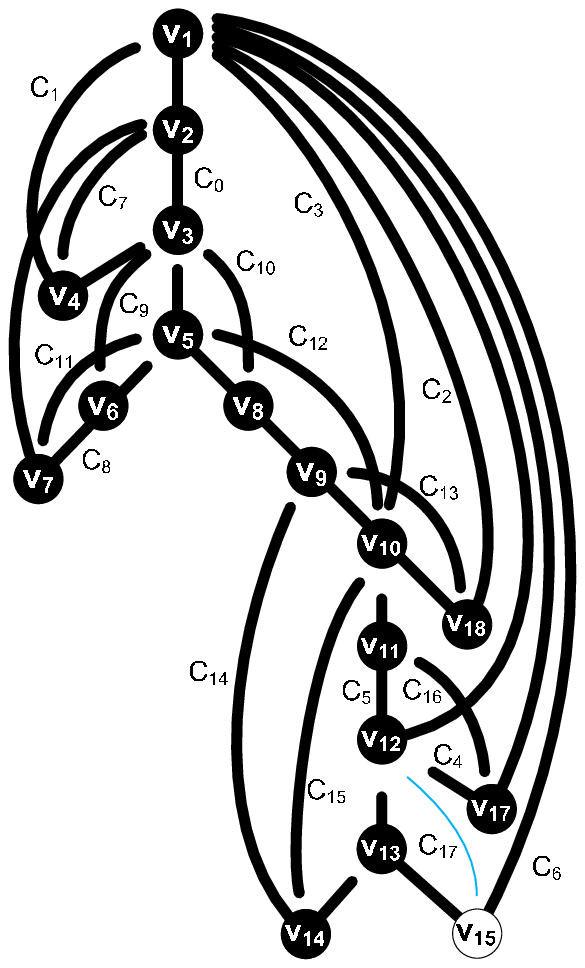}
		\label{fig:Example6}
	}
	\captcont{A running example of the algorithm.}
	\label{fig:Example}
\end{figure}


\paragraph{A New Certifying 3-Connectivity Test.}\label{test}
It remains to deal with the case when the input graph $G$ is not $3$-connected. For simplicity, we assume $G$ to be $2$-connected, although the chain decomposition can check this fact. If $G$ is not $3$-connected, the described algorithm fails to add a \BG-path due to Theorem~\ref{characterizations} when processing some chain, say $C_i$. Therefore, after the processing phase for $C_i$, $\Children(C_i)$ must still contain a chain $C_j$. Let $H$ be the segment containing $C_j$ and let $H' \supseteq H$ be the set of segments that map to the connected component of the interval overlap graph containing $I(H)$. Then the union of dependent paths of the segments in $H'$ is a path $P \subseteq C_i$ and the two extremal attachment points on $P$ of segments in $H'$ build a separation pair. This pair certifies that $G$ is not $3$-connected and can be computed in linear time.


\newpage
\begin{appendix}
\section{Omitted Proofs}
We give the omitted proofs and the preparatory lemmas that lead to them.

\begin{lemma}[aka \textbf{Lemma 8}]
Computing a chain decomposition of a $3$-connected graph and classifying each chain with Algorithm~\ref{alg:classify} takes running time $O(n+m)$.
\end{lemma}
\begin{proof}
The DFS tree $T$ can be obtained in time $O(n+m)$. The subdivision $S_3$ can be found in time linearly dependent on $E(S_3)$ by taking two arbitrary backedges $ra$ and $rb$ with $r$ being the root of $T$ and finding the lowest common ancestor of $a$ and $b$ by traversing $T$ upwards. The computation of each remaining chain $C_i$, $i>2$, takes time linearly dependent on its length, too, which gives a running time of $O(n+m)$ for the chain decomposition.

In order to obtain a fast classification in Algorithm~\ref{alg:classify}, we store the following information on each chain $C_i$: A pointer to its parent $C_k$ (for $C_i \neq C_0$), pointers to $s(C_i)$ and $t(C_i)$ and the information whether $C_i$ is a backedge. In addition, for each inner vertex of $C_i$ a pointer to $C_i$ is stored. That allows us to check vertices on being contained as inner vertices or end vertices in arbitrary chains in $O(1)$. If $C_k=C_0$, we can check the condition on $C_i$ being of type~$1$ in constant time by testing whether $s(C_i)$ and $t(C_i)$ are contained in $C_0$. If $C_k \neq C_0$, we check in constant time whether $s(C_i)$ and $t(C_i)$ are contained in $C_k \setminus \{s(C_k)\}$. The condition for type~$2$ needs constant time as well. Every chain is marked at most once, therefore unmarked as most once in line~\ref{unmarking} of Algorithm~\ref{alg:classify}, which gives a total running time of $O(n+m)$.
\end{proof}

\begin{lemma}\label{whenupwardsclosedandmodular}
Let $S_l$ be upwards-closed and modular. Then a \BG-path $P$ for $S_l$ is a chain if and only if $S_{l+1}$ is upwards-closed and modular.
\end{lemma}
\begin{proof}
If $P$ is a chain, $t(P)$ is contained in $S_l$ and $S_{l+1}$ must be upwards-closed and modular due to the DFS structure. If $P$ is not a chain, we assume to the contrary that $S_{l+1}$ is upwards-closed and modular. Then $P$ must be the union of $t > 1$ chains; let $C_i$ be the first chain in $P$. Now $P$ cannot start with $t(C_i)$, since $s(C_i)$ is in $S_l$ and property~\ref{bgpathdefinition}.\ref{bgpathdefinition1} contradicts $t > 1$. Thus, $P$ starts with $s(C_i)$, which contradicts $t > 1$ as well, as $S_{l+1}$ is upwards-closed and a second chain in $P$ would include another backedge in $P$ at a vertex that is already incident to two DFS tree edges.
\end{proof}

Lemma~\ref{whenupwardsclosedandmodular} shows that this restriction implies every \BG-path to be a chain.

\begin{lemma}\label{restriction}
Each path $P$ in $S^R_l$ having properties~\ref{bgpathdefinition}.\ref{bgpathdefinition1} and~\ref{bgpathdefinition}.\ref{bgpathdefinition2} is a \BG-path. If $P$ is additionally a chain of type~$2a$ or $3a$,~\ref{R1} and~\ref{R2} are preserved.
\end{lemma}
\begin{proof}
For the first claim, assume to the contrary that $P$ violates property~\ref{bgpathdefinition}.\ref{bgpathdefinition3}. Then $|V_{real}(S_l)| \geq 4$ must hold and $S_l \neq S_3$ follows. Let $R$ and $Q$ be the parallel links of $S_l$ that contain the end vertices of $P$ as inner vertices, respectively. At least one of them, say $R$, contains a backedge, since otherwise $T$ would contain a cycle. Let $C_i \neq C_0$ be the chain in $S_l$ that contains $R$. Since $C_i$ contains exactly one backedge, $s(C_i)$ is an end vertex of $R$. If $R \subset C_i$, $Q$ must contain a backedge, as $t(C_i)$ is an inner real vertex of $t(R) \rightarrow_T s(R)$.
In that case, all inner vertices of $Q$ lie in a subtree of $T$ that cannot be reached by $P$ due to property~\ref{bgpathdefinition}.\ref{bgpathdefinition1} and $S_l$ being upwards-closed. Thus, $R = C_i$ and with the same argument $Q = t(C_i) \rightarrow_T s(C_i)$ holds. With~\ref{R2}, $S_l$ must be $S_3$ and $Q = C_0$, which contradicts our assumption.

\begin{figure}
	\centering
	\subfigure[allowed]{
		\includegraphics[scale=0.65]{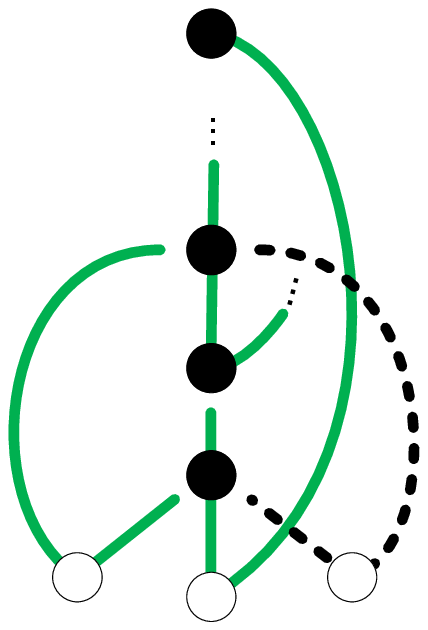}
		\label{fig:Restriction2Allowed}
	}
	\subfigure[forbidden]{
		\includegraphics[scale=0.65]{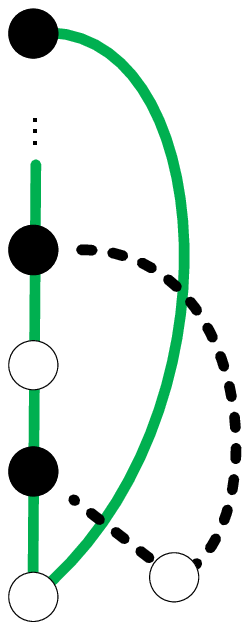}
		\label{fig:Restriction2Forbidden}
	}
	\caption{The effect of restriction~\ref{R2} on the dashed \BG-path.}
	\label{fig:Restriction2}
\end{figure}

For the second claim, each chain $C_i$ of type~$2$ or $3$ has by definition an inner real vertex in $t(C_i) \rightarrow_T s(C_i)$ and therefore preserves~\ref{R2}. If $C_i$ is of type~$2a$ or $3a$, \ref{R1} is preserved as well, as $S_{l+1}$ is upwards-closed and modular with Lemma~\ref{whenupwardsclosedandmodular}.
\end{proof}

\begin{figure}
	\centering
	\includegraphics[scale=0.65]{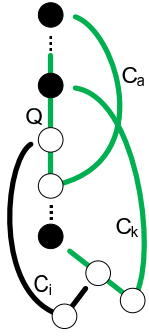}
	\caption{A chain $C_i \not \subseteq S_l$ of type~$3$.}
	\label{fig:type3chain}
\end{figure}

We show that chains of type~$3a$ help to find \BG-paths efficiently (\textbf{Lemma~10}) as part of the following Lemma.

\begin{lemma}[aka \textbf{Lemma~10}]\label{structural}
Let $C_k$ be the parent of a chain $C_i \neq C_0$.
\begin{itemize}
	\item[--] If $C_i$ is not of type~$1$, $C_k \neq C_0$ and $t(C_i)$ is an inner vertex of $C_k$.
	\item[--] Let $C_k$ but not $C_i$ be contained in $S^R_l$. If $C_i$ is either of type~$1$ with an inner real vertex in $t(C_i) \rightarrow_T s(C_i)$ or of type~$3a$, $C_i$ is a \BG-path for $S^R_l$ preserving~\ref{R1} and~\ref{R2}.
\end{itemize}
\end{lemma}
\begin{proof}
Assume to the contrary that $C_i$ is not of type~$1$ and $C_k = C_0$. Because $t(C_i)$ is contained in $C_0$, $s(C_i)$ must be in $C_0$ as well. But then $C_i$ would be of type~$1$, since $t(C_i) \rightarrow_T s(C_i) \subseteq C_0$. Therefore, if $C_i$ is not of type~$1$, $C_k \neq C_0$ holds and $C_k$ must start with a backedge. Then the definition of the parent relation implies that $t(C_i)$ is an inner vertex of $C_k$.

For the second claim, let $C_i$ first be of type~$3a$. Since $S_l$ is upwards-closed, modular and contains $C_k$, $C_i$ satisfies the property~\ref{bgpathdefinition}.\ref{bgpathdefinition1} of \BG-paths. In addition, $s(C_i) \neq s(C_k)$ holds by definition and with $C_k < C_i$, $s(C_i)$ must be an inner vertex of the path $t(C_k) \rightarrow_T s(C_k)$ (see Figure~\ref{fig:type3chain}). Therefore, the only chains $C_j$ that contain $s(C_i)$ and $t(C_i)$ are different from $C_0$ and fulfill $C_i \cap C_j = \{s(C_i),t(C_i)\} = \{s(C_j),t(C_j)\}$. This implies $C_i$ having property~\ref{bgpathdefinition}.\ref{bgpathdefinition2}. Using Lemma~\ref{restriction}, $C_i$ is a \BG-path for $S_l$ that preserves~\ref{R1} and~\ref{R2}.

If $C_i$ is of type~$1$, property~\ref{bgpathdefinition}.\ref{bgpathdefinition1} follows from the same argument as before. Additionally, the inner real vertex in $t(C_i) \rightarrow_T s(C_i)$ prevents any link containing $s(C_i)$ and $t(C_i)$ from having $s(C_i)$ or $t(C_i)$ as an inner vertex and therefore ensures property~\ref{bgpathdefinition}.\ref{bgpathdefinition2}. Lemma~\ref{restriction} implies that $C_i$ is a \BG-path for $S_l$ and $C_i$ must preserve~\ref{R1} and~\ref{R2}, the latter due to the inner real vertex in $t(C_i) \rightarrow_T s(C_i)$.
\end{proof}

\begin{proposition}\label{caterpillar3b}
Every caterpillar $L_j$ consists of exactly one chain of type~$3b$, namely the chain $C_j$, and one or more chains of type~$2b$.
\end{proposition}

\begin{lemma}\label{caterpillarPartition}
$C \setminus \{C_0\}$ is partitioned into single chains of types $1$, $2a$, $3a$ and the chains being contained in caterpillars. Moreover, no chain is contained in two caterpillars.
\end{lemma}
\begin{proof}
With Proposition~\ref{caterpillar3b}, it remains to show that every chain $C_i$ of type~$2b$ or $3b$ is contained in exactly one caterpillar. If $C_i$ is of type~$3b$, $C_i$ is part of the caterpillar $L_i$ (see Algorithm~\ref{alg:classify}, line~\ref{assign}) and will not be assigned to a second caterpillar afterwards, as it is not marked. Otherwise, $C_i$ is of type~$2b$ and was therefore marked. We show that, after all chains in $C$ have been classified, $C_i$ is not marked anymore. This forces $C_i$ to be contained in exactly one caterpillar, as the only way to unmark chains is to assign them to a caterpillar (see Algorithm~\ref{alg:classify}, line~\ref{unmarking}) and no chain is marked twice.

Let $C_k$ be the parent of $C_i$. Because $C_i$ is of type~$2b$, $s(C_i)=s(C_k)$ holds and $C_i$ is not a backedge, implying that the last edge $e$ of $C_i$ is in $T$. Let $x$ be the end vertex of $e$ different from $t(C_i)$. Using Lemma~\ref{structural}, $C_k \neq C_0$ holds and $t(C_i)$ is an inner vertex of $C_k$. Then at least one backedge $vw$ with $v \notin \{s(C_i),t(C_i)\}$ must enter $T(x)$, since otherwise $s(C_i)$ and $t(C_i)$ would be a separation pair of $G$. Let $C_j$ be the minimal chain with respect to $<$ that contains such a backedge.

As $C_j > C_i$ holds
, the vertex $v$ is an inner vertex of $t(C_i) \rightarrow_T s(C_i)$, implying that $C_j$ is not of type~$2$. In addition, $C_j$ is not of type~$1$, since $t(C_j) \rightarrow_T v$ contains edges from $C_i$ and $C_k$. At the time $C_j$ is found in the chain decomposition, every chain that already ends at a vertex in $T(x)$ starts at $s(C_i)$ and is therefore of type~$2a$ or $2b$. Since chains that are backedges cannot have children, the parent of $C_j$ is marked and $C_j$ is of type~$3b$. Moreover, every chain corresponding to an inner vertex of the path $C_j \rightarrow_U C_i$ is marked. This concludes $C_i$ to become unmarked due to line~\ref{unmarking} of Algorithm~\ref{alg:classify}.
\end{proof}

The following gives the detailed proof of \textbf{Lemma~12}.

\begin{lemma}[aka \textbf{Lemma~12}]
Let $L_j$ be a caterpillar that consists of $t$ chains and has parent $C_k$. Let $C_k$ but no chain in $L_j$ be contained in $S^R_l$. Then, if and only if $L_j$ is good, the chains in $L_j$ can be efficiently decomposed into $t$ successive \BG-paths satisfying~\ref{R1} and creating subdivisions $S_{l+1},S_{l+2},\ldots,S_{l+t-1},S^R_{l+t}$, each of which satisfies~\ref{R2}.
\end{lemma}
\begin{proof}
Let $L_j$ be good and let $y$ be the last vertex of the minimal chain in $L_j$, thus $y \in V(C_k)$. We assume at first that $s(C_j)$ is a proper ancestor of $t(C_k)$ (see Figure~\ref{fig:CaterpillarGood}). Then the path $P = C_j \cup (t(C_j) \rightarrow_T y)$ fulfills properties~\ref{bgpathdefinition}.\ref{bgpathdefinition1} and \ref{bgpathdefinition}.\ref{bgpathdefinition2} and is a \BG-path for $S_l$ with Lemma~\ref{restriction}. Adding $P$ preserves $S_l$ to be upwards-closed but not modular. The restriction~\ref{R2} is also preserved, as $t(C_k)$ is real and, for $S_l = S_3$, $C_k$ must be either $C_1$ or $C_2$, implying that $s(P)$ becomes an inner real vertex of $C_0$. Successively, for each chain $C_i$ of the $t-1$ chains in $L_i \setminus \{C_j\}$, we now add $C_i \setminus P$, which is a \BG-path yielding an upwards-closed subdivision for analogue reasons.

Now assume that $s(C_j)$ is a descendant of $t(C_k)$ (see Figure~\ref{fig:CaterpillarBad}). Then $s(C_j) \in V(C_k)$ and since $L_j$ is good, there is a real vertex $a$ strictly between $s(C_j)$ and $s(C_k)$ in $C_k$. We first show that $t(C_k) \rightarrow_T s(C_k)$ contains an inner real vertex as well. Assume the contrary. Then $C_k$ must be of type~$1$ and has been added before, contradicting restriction~\ref{R2} unless $S_l = S_3$. But $S_l$ must be different from $S_3$, since $a$ exists, and it follows that $t(C_k) \rightarrow_T s(C_k)$ contains an inner real vertex $b$. Let $C_h$ be the parent of $C_j$. Then $(C_j \cup C_h) \setminus ((t(C_j) \rightarrow_T y) \setminus \{t(C_j)\})$ is a \BG-path due to the real vertices $a$ and $b$ and we add it, although it neither preserves $S_{l+1}$ to be upwards-closed nor modular. We next add $t(C_j) \rightarrow_T y$, which restores upwards-closedness. The resulting subdivisions $S_{l+1}$ and $S_{l+2}$ both satisfy~\ref{R2}, as $b$ is real in $S_{l+1}$ and $S_{l+2}$ and $y$ is real in $S_{l+2}$. We proceed with adding successively paths, namely for each chain $C_i$ of the $t-2$ remaining chains in $L_i \setminus \{C_j,C_h\}$ the path $C_i \setminus (t(C_j) \rightarrow_T y)$. With the same line of argument, these paths obtain upwards-closed subdivisions $S_{l+3},\ldots,S_{l+t}$, each of which satisfies~\ref{R2}.

In both cases, $S_{l+t}$ is modular, since $L_j$ is a list of chains. Moreover, the $t$ chosen \BG-paths preserve~\ref{R1}, as the chains in $L_j$ are of types $2b$ and $3b$ only, $t > 1$ holds and $S_{l+t}$ is upwards-closed. All paths can be computed in time linearly dependent on the total number of edges in $L_j$.

For the only if part, let $P_1$ and $P_2$ be the first two \BG-paths in a decomposition of the chains in $L_j$; these exist, since $t > 1$ holds in every caterpillar. Let $L_j$ be bad, as otherwise the claim follows. Then $s(C_j) \in V(C_k)$. We show that $L_j$ cannot be bad, as $S_l$ contains a real vertex in $C_k$ strictly between $s(C_j)$ and $s(C_k)$. Because of properties~\ref{bgpathdefinition}.\ref{bgpathdefinition1} and \ref{bgpathdefinition}.\ref{bgpathdefinition2}, $P_1 \cap S_l = \{s(C_k),s(C_j)\}$ must hold and $P_1$ is a link of $S_{l+1}$ being parallel to $s(C_j) \rightarrow_{C_k} s(C_k)$. Since only the chain of type~$3b$ in $L_j$ starts at $s(C_j)$, both end vertices of $P_2$ must be different from $s(C_j)$. Then, due to properties~\ref{bgpathdefinition}.\ref{bgpathdefinition1} and~\ref{bgpathdefinition}.\ref{bgpathdefinition2}, $P_2$ joins inner vertices of the parallel links $P_1$ and $s(C_j) \rightarrow_{C_k} s(C_k)$ in $S_{l+1}$, contradicting property~\ref{bgpathdefinition}.\ref{bgpathdefinition3}, as $|V_{real}(S_{l+1})| \geq 4$.
\end{proof}

\begin{lemma}\label{3a3b}
Let $C_i$ be a chain of type~$3$ such that $s(C_i) \in V(S^R_l)$, $C_i \not \subseteq S^R_l$ and $C_i$ is minimal among the chains of type~$3$ in its segment $H$. Let $D_1 > \ldots > D_k$ be all ancestors of $C_i$ in $H$ with $D_1 = C_i$ and $D_k$ being the minimal chain in $H$. Then all chains $D_k,\ldots,D_1$ can be successively added as \BG-paths preserving~\ref{R1} and~\ref{R2} (possibly being part of caterpillars), unless one of the following exceptions holds:

\begin{enumerate}
	\item $C_i$ is of type~$3a$, $k=2$, $D_k$ is of type~$1$, $s(C_i)$ is an inner vertex of $t(D_k) \rightarrow_T s(D_k)$ and there is no inner real vertex in $t(D_k) \rightarrow_T s(D_k)$ (Figure~\ref{fig:3a3bOne}),\label{3a3bOne}
	\item $C_i$ is of type~$3b$, $D_k$ is of type~$2b$ and $L_i = \{D_1,\ldots,D_k\}$ with $L_i$ being bad (Figure~\ref{fig:3a3bTwo}),\label{3a3bTwo}
	\item $C_i$ is of type~$3b$, $L_i = \{D_1,\ldots,D_{k-1}\}$, $D_k$ is of type~$1$, $s(C_i)$ is an inner vertex of $t(D_k) \rightarrow_T s(D_k)$ and there is no inner real vertex in $t(D_k) \rightarrow_T s(D_k)$ (Figure~\ref{fig:3a3bThree}).\label{3a3bThree}
\end{enumerate}
\end{lemma}

\begin{figure}
	\centering
	\subfigure[\ref{3a3b}.\ref{3a3bOne}]{
		\includegraphics[scale=0.68]{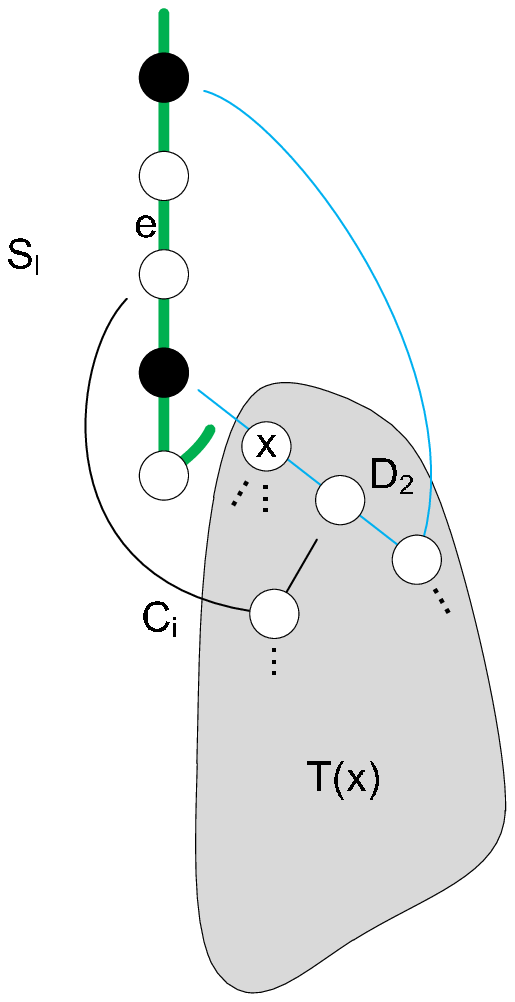}
		\label{fig:3a3bOne}
	}
	\hfill
	\subfigure[\ref{3a3b}.\ref{3a3bTwo}]{
		\includegraphics[scale=0.68]{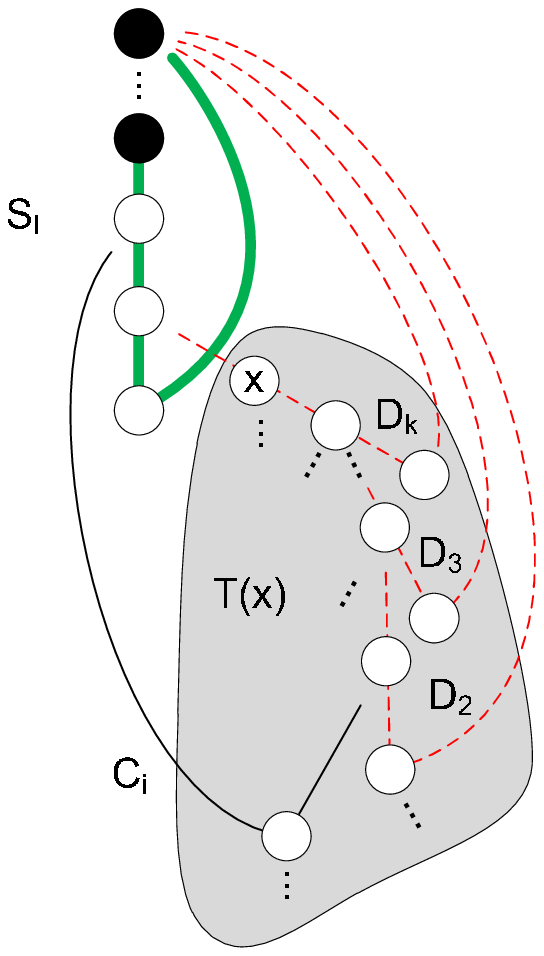}
		\label{fig:3a3bTwo}
	}
	\hfill
	\subfigure[\ref{3a3b}.\ref{3a3bThree}]{
		\includegraphics[scale=0.68]{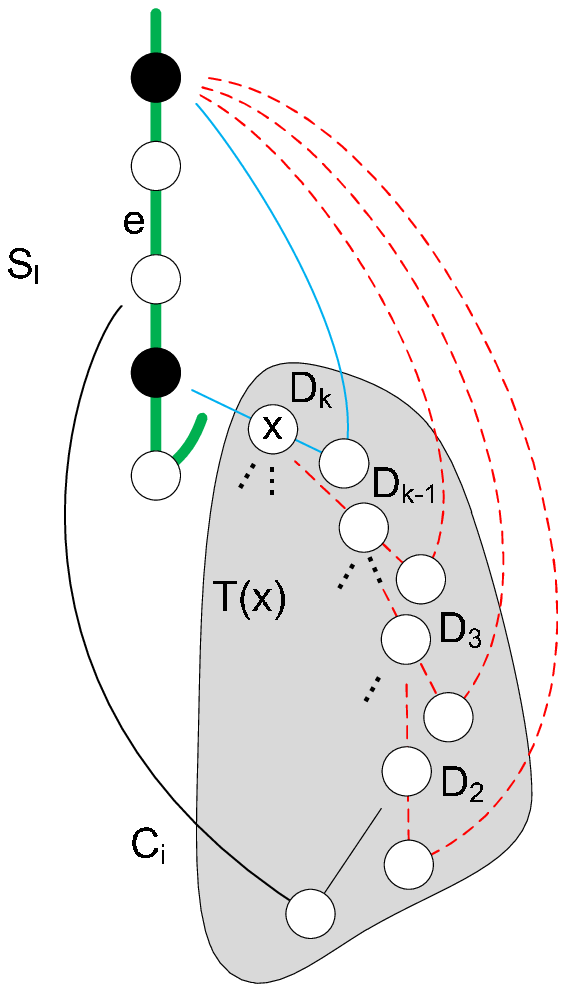}
		\label{fig:3a3bThree}
	}
	\caption{The three exceptions of Lemma~\ref{3a3b}. The black vertices in~\ref{3a3b}.\ref{3a3bOne} and~\ref{3a3b}.\ref{3a3bThree} can also be non-real.}
	\label{fig:3a3b}
\end{figure}

\begin{proof}
Let $D \in \{D_2,\ldots,D_k\}$. Then $D$ is not of type~$3$ by assumption and not of type~$2a$, as chains of that type cannot have children. Assume that $D$ is of type~$2b$ and let $L_j$ be the caterpillar containing $D$ due to Lemma~\ref{caterpillarPartition}. If $C_j \neq C_i$, $C_j < C_i$ holds, as otherwise $C_j$ would not be the chain of type~$3b$ in $L_j$. But then $C_j$ contradicts the minimality of $C_i$, since $C_j$ is not contained in $S_l$ and of type~$3b$. We conclude that every chain in $\{D_2,\ldots,D_k\}$ of type~$2b$ is contained in $L_i$ and forces $C_i$ to be of type~$3b$. This is used in the following case distinction.

Let $C_i$ be of type~$3a$. If $k=1$, $C_i$ is a \BG-path for $S_l$ with Lemma~\ref{structural} and the claim follows. Otherwise, $k > 1$ and all chains in $\{D_2,\ldots,D_k\}$ are of type~$1$. Then $s(D_2)$ is a proper ancestor of $s(C_i)$, since $D_2 < C_i$ and $C_i$ is not of type~$2$. Moreover, $s(C_i)$ is a proper ancestor of $t(D_2)$, because otherwise $H \cap S_l = \{s(D_2),t(D_2)\}$ is a separation pair of $G$ due to the minimality of $C_i$. It follows that $s(C_i)$ is an inner vertex of $t(D_2) \rightarrow_T s(D_2)$. If $k > 2$, $D_3$ must contain $t(D_2) \rightarrow_T s(D_2)$, because $D_2$ is of type~$1$ and a child of $D_3$. Therefore, the edge $e$ joining $s(C_i)$ with the parent of $s(C_i)$ in $T$ is contained in $D_3$. But since $S_l$ is upwards-closed, $e$ is also contained in $S_l$, contradicting that $D_3 \not \subseteq S_l$. Thus, $k=2$. If $t(D_2) \rightarrow_T s(D_2)$ contains an inner real vertex, $D_2$ and $C_i$ can be subsequently added as \BG-paths with Lemma~\ref{structural}, otherwise~\ref{3a3b}.\ref{3a3bOne} is satisfied.

Let $C_i$ be of type~$3b$. Then all chains in $\{D_2,\ldots,D_k\}$ that are of type~$2b$ must be contained in $L_i$. Since every caterpillar $L_j$ contains the parent of the chain $C_j$ and since $S_l$ contains no chain in $L_i$ due to~\ref{R1}, $k > 1$ holds and $D_2$ is of type~$2b$ with $D_2 \in L_i$. Let $D_t$ with $1 < t \leq k$ be the minimal chain in $L_i$. If $t=k$ and $L_i$ is good, all chains in $L_i$ can be decomposed to \BG-paths according to Lemma~\ref{AddCaterpillar}. If $t=k$ and $L_i$ is bad, \ref{3a3b}.\ref{3a3bTwo} is satisfied. Only the case $k > t$ remains. Then $D_{t+1}$ is of type~$1$ and, using the same arguments as in the case for type~$3a$, $s(C_i)$ is an inner vertex of $t(D_k) \rightarrow_T s(D_k)$ and $k = t+1$. If $t(D_k) \rightarrow_T s(D_k)$ contains an inner real vertex, Lemmas~\ref{structural} and~\ref{AddCaterpillar} imply that $D_k$ and $L_i$ can be iteratively added as set of successive \BG-paths, preserving~\ref{R1} and~\ref{R2}. Otherwise,~\ref{3a3b}.\ref{3a3bThree} is satisfied.
\end{proof}

We extend Lemma~\ref{3a3b} to non-minimal chains of type~$3$.

\begin{figure}
	\centering
	\subfigure[$C_j$ cannot be contained in exceptions~\ref{3a3b}.\ref{3a3bOne} and~\ref{3a3b}.\ref{3a3bThree}.]{
		\includegraphics[scale=0.7]{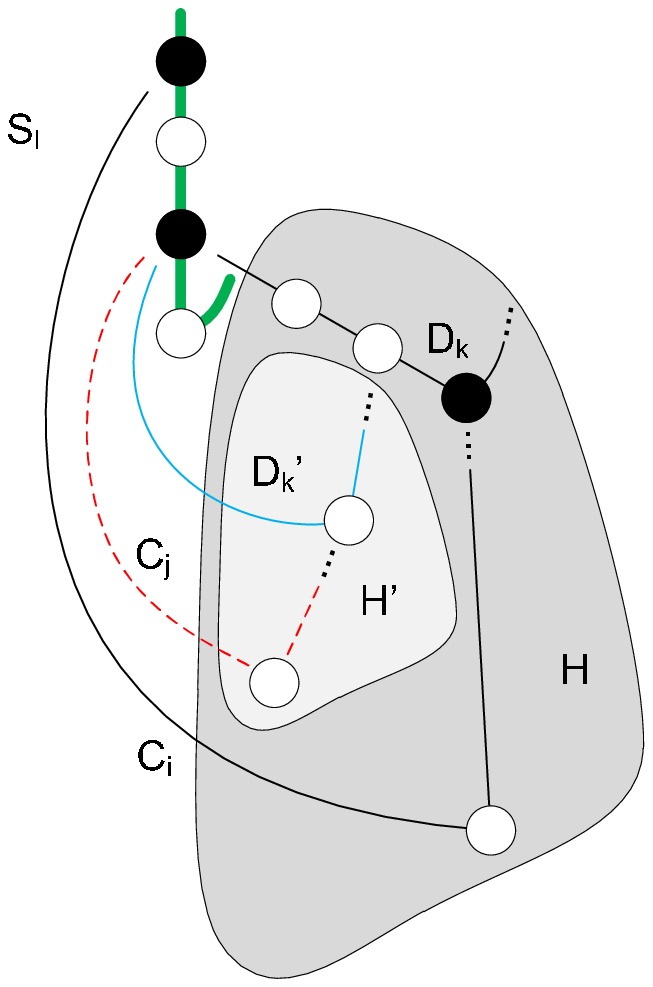}
		\label{fig:All3a3bOneThree}
	}
	\hspace{0.2cm}
	\subfigure[$C_j$ cannot be contained in exception~\ref{3a3b}.\ref{3a3bTwo}.]{
		\includegraphics[scale=0.7]{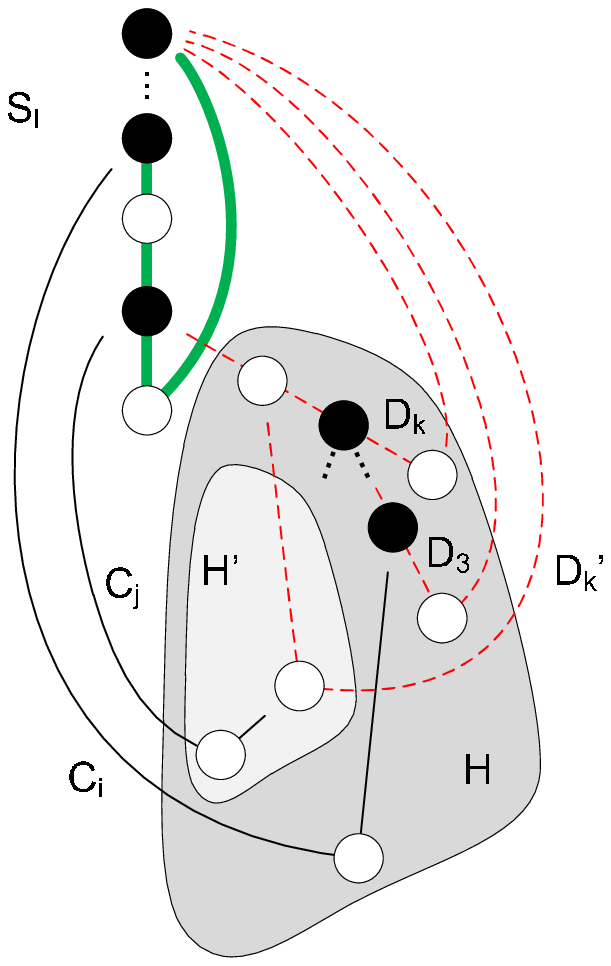}
		\label{fig:All3a3bTwo}
	}
	\caption{After $C_i$ was added, the next minimal chain $C_j$ is in no exception.}
	\label{fig:all3a3b}
\end{figure}

\begin{lemma}\label{all3a3b}
Let the preconditions of Lemma~\ref{3a3b} hold. If $C_i$ is not contained in one of the exceptions~\ref{3a3b}.\ref{3a3bOne}-\ref{3a3b}.\ref{3a3bThree} (as $C_i$), the chains of type~$3$ in $H$ that start in $S^R_l$ and their ancestors in $H$ can be successively added as \BG-paths (in reversed order), preserving~\ref{R1} and~\ref{R2}.
\end{lemma}
\begin{proof}
Using Lemma~\ref{3a3b}, we add the chains $C_i,D_2,\ldots,D_k$ in $H$ as \BG-paths. This partitions $H$ into new segments; let $H' \subseteq H \setminus \{C_i,D_2,\ldots,D_k\}$ be such a new segment. If $H'$ does not contain chains of type~$3$ that start in $S_l$, the claim follows for such chains in $H'$. Otherwise, let $C_j$ be the minimal chain of type~$3$ in $H'$ that starts in $S_l$ and let $C_j > D'_2 > \ldots > D'_k$ be its ancestors in $H'$. We show that $C_j$ is not contained in one of the exceptions~\ref{3a3b}.\ref{3a3bOne}-\ref{3a3b}.\ref{3a3bThree} and can therefore be added as \BG-path with Lemma~\ref{3a3b}, along with its proper ancestors in $H'$. First, assume to the contrary that $C_j$ is contained in exception~\ref{3a3b}.\ref{3a3bOne} or~\ref{3a3b}.\ref{3a3bThree} (see Figure~\ref{fig:All3a3bOneThree}). Because $D'_k$ is a proper descendant of $D_k$ and $D'_k$ is of type~$1$, $s(C_j) \in V(S_l)$ cannot be an inner vertex of $t(D'_k) \rightarrow_T s(D'_k)$, contradicting the assumption. Now assume to the contrary that $C_j$ is contained in exception~\ref{3a3b}.\ref{3a3bTwo} (see Figure~\ref{fig:All3a3bTwo}). Then $C_j$ is of type~$3b$ and part of a bad caterpillar $L_j$, whose parent $D$ is not contained in $H'$. Because $L_j$ contains only chains in $H'$, $D$ must be a descendant of $D_k$ and is therefore contained in $H \setminus H'$. Since $L_j$ is bad, $s(C_j)$ is contained in $S_l \cap D$ and it follows with $s(C_j) \neq s(D)$ that $D$ must end in $S_l$ at the vertex $s(C_j)$. As $D_k$ is the only chain in $H$ that ends in $S_l$, $D=D_k$ must hold. But this contradicts $L_j$ being bad, as $D$ contains the inner real vertex $s(D_{k-1})$. Thus, $C_j$ and its ancestors in $H'$ can be added, partitioning $H'$ into smaller segments. Iterating the same argument for these segments establishes the claim for all chains of type~$3$ in $H$ that start in $S_l$.
\end{proof}

The next lemma shows that the only chains of type~$1$ that cannot be added are either backedges or are contained as $D_k$ in exceptions~\ref{3a3b}.\ref{3a3bOne}-\ref{3a3b}.\ref{3a3bThree}.

\begin{lemma}\label{backedgeOr}
Let $C_j$ be a chain in $S^R_l$ and let $D_k$ be a child of $C_j$ that is of type~$1$ and not in $S^R_l$. If $D_k$ is not a backedge, there is a chain of type~$3$ in the segment containing $D_k$ that starts in $t(D_k) \rightarrow_T s(D_k) \subset C_j$. If $D_k$ is neither a backedge nor contained (as $D_k$) in the exceptions~\ref{3a3b}.\ref{3a3bOne} and~\ref{3a3b}.\ref{3a3bThree}, $D_k$ can be added as \BG-path.
\end{lemma}
\begin{proof}
Let $H$ be the segment of $D_k$ and assume that $D_k$ is not a backedge. We first show that $H$ contains a chain of type~$3$ that starts in $t(D_k) \rightarrow_T s(D_k)$. We can assume that $D_k$ is not contained in the exceptions~\ref{3a3b}.\ref{3a3bOne} and~\ref{3a3b}.\ref{3a3bThree}, as then $H$ would contain such a chain by definition. Let $x$ be the last but one vertex in $D_k$. Since $G$ is $3$-connected, there is a minimal chain $C_i$ entering $T(x)$ such that $s(C_i)$ is an inner vertex of $t(D_k) \rightarrow_T s(D_k)$, as otherwise the inner vertices of $D_k$ would be separated by $\{s(D_k),t(D_k)\}$. By definition of the chain decomposition, $C_i$ must be of type~$3a$ or~$3b$. Because $D_k$ is not contained in exceptions~\ref{3a3b}.\ref{3a3bOne} and~\ref{3a3b}.\ref{3a3bThree} and $H$ cannot contain exception~\ref{3a3b}.\ref{3a3bTwo}, Lemma~\ref{3a3b} can be applied on $C_i$, obtaining the last claim.
\end{proof}

Here we give the proof of the key theorem (\textbf{Theorem~15}) of the paper.

\begin{theorem}[aka \textbf{Theorem~15}]
For a subdivision $S^R_l$, let $C_i$ be a chain such that $\Children(C_j) = \Type(C_j) = \emptyset$ holds for every proper ancestor $C_j$ of $C_i$. Then all chains in $\Children(C_i) \cup \Type(C_i)$ and their proper ancestors that are not already contained in $S^R_l$ can be successively added as \BG-paths (possibly being part of caterpillars) such that~\ref{R1} and~\ref{R2} is preserved. Moreover, the chains in $\Type(C_i)$ that are contained in segments in which the minimal chain is not contained in $\Children(C_i)$ can be added at any point in time in arbitrary order (together with their proper ancestors that are not contained in $S^R_l$).
\end{theorem}
\begin{proof}
By assumption, $C_i$ is contained in $S_l$. Let $D \not \subseteq S_l$ be a child of $C_i$. If $C_i = C_0$, $D$ must be of type~$1$. Let $C_i \neq C_0$. Then $D$ cannot be of type~$3b$, as otherwise it would be contained in $S_l$ due to~\ref{R1} and $C_i \subset S_l$. It is neither of type~$3a$, since in that case $s(D)$ is contained in a proper ancestor of $C_i$, implying $D \subset S_l$ by assumption. We conclude that $D$ is of type~$1$ or~$2$ and focus on the cases where $D$ can not be added. Let $P$ be the path on which $D$ depends on. If $D$ is of type~$1$, $P$ does not contain an inner real vertex, as otherwise $D$ can be added as \BG-path due to Lemma~\ref{structural}. With Lemma~\ref{backedgeOr}, $D$ must be either a backedge or be contained as the minimal chain in exception~\ref{3a3b}.\ref{3a3bOne} or~\ref{3a3b}.\ref{3a3bThree}. If $D$ is of type~$2a$, $s(C_i)$ is real and neither $t(D)$ nor an inner vertex in $P$ can be real, since otherwise $D$ can be added as \BG-path, preserving~\ref{R1} and~\ref{R2}. If $D$ is of type~$2b$, $D$ is the minimal chain of a caterpillar $L_a$ with parent $C_i$. According to Lemma~\ref{AddCaterpillar}, $L_a$ is bad and, thus, corresponds to exception~\ref{3a3b}.\ref{3a3bTwo}. The following is a list of the possible cases for which a child $D$ of $C_i$ is not added.

\begin{enumerate}
	\item $D$ is of type~$1$ without an inner real vertex in $P$ and either a backedge or the minimal chain in exception~\ref{3a3b}.\ref{3a3bOne} or~\ref{3a3b}.\ref{3a3bThree}\label{correctnessOne}
	\item $C_i \neq C_0$ and $D$ is of type~$2a$ without a real vertex in $P \setminus \{s(D)\}$\label{correctnessTwo}
	\item $C_i \neq C_0$ and $D$ is of type~$2b$ without an inner real vertex in $P$ ($D$ is the minimal chain in exception~\ref{3a3b}.\ref{3a3bTwo})\label{correctnessThree}
\end{enumerate}

We iteratively add all chains $D$ in $X'\ \dot{\cup}\ Y'$ that do not satisfy one of the above three cases~\ref{correctness}.\ref{correctnessOne}-\ref{correctness}.\ref{correctnessThree} for $D \in X'$ and whose segments do not contain one of the exceptions~\ref{3a3b}.\ref{3a3bOne}-\ref{3a3b}.\ref{3a3bThree} for $D \in Y'$ (the latter followed by adding the proper ancestors in the segment of $D$ according to Lemma~\ref{all3a3b}). Let $X$ be the set of remaining chains in $X'$ and let $Y$ be the set of remaining chains in $Y'$. If $X = \emptyset$, $Y = \emptyset$ holds as well, as otherwise the minimal chain in the segment containing one of the exceptions~\ref{3a3b}.\ref{3a3bOne}-\ref{3a3b}.\ref{3a3bThree} is a child of $C_i$, contradicting $X = \emptyset$. This implies the claim for $X = \emptyset$.

We prove the theorem by showing that $X = \emptyset$ must hold. Assume to the contrary that $X \neq \emptyset$ and let $S_t$ be the current subdivision (all segments will be dependent on $S_t$). Then $C_i$ must contain a link $L$ of length at least two, because the dependent path $P$ in each of the cases~\ref{correctness}.\ref{correctnessOne}-\ref{correctness}.\ref{correctnessThree} is in $C_i$ and contains a non-real vertex due to the $3$-connectivity and simpleness of $G$. According to Lemma~\ref{multipleconstruction}, $L$ contains an inner vertex $v$ on which a \BG-path $B$ starts (not necessarily being a chain and not necessarily preserving~\ref{R1} or~\ref{R2}). Let $e$ be the first edge of $B$. Then $e$ is not contained in the segment of any $x \in X$, as otherwise $B$ would not have property~\ref{bgpathdefinition}.\ref{bgpathdefinition2}, because $v$ is non-real and all start vertices of the chains in the segment of $x$ that are in $S_t$ are contained in $L$. Thus, $C(e)$ cannot be a child of $C_i$ and it follows that $s(C(e))=v$. In particular, $C(e)$ is not of type~$1$.

The segment of $e$ cannot contain a chain of type~$3$ that starts in $C_i$, as it otherwise contains a chain $x \in X$ of type~$1$ or~$2b$ due to exceptions~\ref{3a3b}.\ref{3a3bOne}-\ref{3a3b}.\ref{3a3bThree}, contradicting the previous argument. In particular, $C(e)$ is not of type~$3$ and the only remaining case is that $C(e)$ is of type~$2$.

Let $C_k$ be the maximal ancestor of $C(e)$ that is not of type~$2$. Then $s(C_k)=v$ holds by construction of the chain decomposition and $C_k$ must be contained in the segment of $C(e)$ due to~\ref{R1}, \ref{R2} and $v$ being non-real. Since the segment of $e$ cannot contain a chain of type~$3$ that starts in $C_i$, $C_k$ must be of type~$1$. But then, as $v$ is an inner vertex of $L$, $C_k$ must be a child of $C_i$, contradicting that $e$ is not contained in the segment of any $x \in X$. This is a contradiction to the existence of $B$ and it follows that $X = \emptyset$, which implies the claim.
\end{proof}
\end{appendix}

\small
\bibliographystyle{abbrv}
\bibliography{../../../../Publications/Jens}

\end{document}